\documentclass[a4paper]{article}

\usepackage[utf8]{inputenc}
\usepackage[english]{babel}

\usepackage{fullpage} 

\usepackage[textsize=scriptsize]{todonotes}

\usepackage{authblk}

\usepackage[shortlabels]{enumitem}

\usepackage{amsmath}
\usepackage{amssymb}
\usepackage{bbm} 
\usepackage{mathrsfs}

\usepackage{url}
\usepackage{hyperref}
\usepackage{bookmark}
\hypersetup{
    colorlinks = true,
    citecolor = blue,
    linkcolor = blue,
    bookmarksnumbered
}

\usepackage{graphicx}
\usepackage[width=.9\textwidth]{caption}
\usepackage{subcaption}
\usepackage{float}

\usepackage{tikz}
\usetikzlibrary{calc,fit,shapes}

\usepackage{cite}

\usepackage[ruled, vlined, linesnumbered]{algorithm2e}
\SetKwFor{RepTimes}{repeat}{times}

\usepackage[noabbrev]{cleveref}

\usepackage{centernot}

\usepackage[misc]{ifsym}

\usepackage[title]{appendix}

\usepackage{amsthm}
\theoremstyle{plain}
    \newtheorem{theorem}{Theorem}
    \newtheorem{lemma}{Lemma}
    
    \newtheorem{claim}{Claim}
    
\theoremstyle{definition}

\crefname{definition}{definition}{definitions}
\crefname{theorem}{theorem}{theorems}
\crefname{corollary}{corollary}{corollaries}
\crefname{lemma}{lemma}{lemmas}
\crefname{proposition}{proposition}{propositions}
\crefname{claim}{claim}{claims}
\crefname{remark}{remark}{remarks}
\usepackage{mathtools}

\newcommand{\D}{\mathcal{D}}
\renewcommand{\S}{\mathcal{S}}
\newcommand{\I}{\mathcal{I}}
\newcommand{\F}{\mathcal{F}}
\newcommand{\E}{\mathbb{E}}
\renewcommand{\P}{\mathbb{P}}
\newcommand{\R}{\mathbb{R}}
\newcommand{\bra}[1]{\left[#1\right]}
\newcommand{\cbra}[1]{\left\{#1\right\}}
\newcommand{\Par}[1]{\left(#1\right)}
\newcommand{\Abs}[1]{\left|#1\right|}
\newcommand{\core}{\text{core}}
\newcommand{\val}{\text{val}}
\newcommand{\SAA}{\text{SAA}}
\newcommand{\1}{\mathbbm{1}}

\def\keywordname{{\bf Keywords:}}
\providecommand{\keywords}[1]{\def\and{{\textperiodcentered} }
\par\addvspace\baselineskip
\noindent\keywordname\enspace\ignorespaces#1}

\title{Two-stage Stochastic Assignment Games}
\author[1]{Laura Sanità}
\author[2]{Lucy Verberk\(^{(\text{\Letter})}\)}
\affil[1]{
    Bocconi University of Milan, Italy. \protect \\
    {\normalsize\tt{laura.sanita@unibocconi.it}}
}
\affil[2]{
    Eindhoven University of Technology, Netherlands. \protect \\
    {\normalsize\tt{l.p.a.verberk@tue.nl}}
}
\date{}

\begin{document}
\maketitle

\begin{abstract}
    In this paper, we study a two-stage stochastic version of the \emph{assignment game}, which is a fundamental cooperative game.
    Given an initial setting, the set of players may change in the second stage according to some probability distribution, and the goal is to find core solutions that are minimally modified.
    
    When the probability distribution is given explicitly, 
    we observe that the problem is polynomial time solvable, as it can be modeled as an LP. More interestingly, we prove that the underlying polyhedron is \emph{integral}, and exploit this in two ways. 
    
    First, integrality of the polyhedron allows us to show that the problem can be well approximated when the distribution is unknown, which is a hard setting.
    
    Second, we can establish an intimate connection to the well-studied multistage \emph{vertex cover} problem. Here, it is known that the problem is NP-hard even when there are only 2 stages and the graph in each stage is bipartite. As a byproduct of our result, we can prove that the problem is polynomial-time solvable if the bipartition is the same in each stage.
    
\end{abstract}

\section{Introduction}

The \emph{Assignment game} is one of the most fundamental cooperative games, introduced in the seminal paper of Shapley and Shubik~\cite{Shapley1971Assignment}. An instance is defined on a bipartite graph \(G = (V,E)\), where the vertices represent players. Each player can participate in one coalition with a subset of other players. If a subset of players \(S \subseteq V\) decides to form a coalition, they can distribute a total value of \(\nu(G[S])\) among them, where \(G[S]\) is the subgraph of \(G\) induced by the vertices in \(S\), and  \(\nu(G)\) denotes the value of a maximum (cardinality) matching in the graph \(G\).
We define an allocation vector \(y \in \mathbb{R}^V_{\geq0}\), where \(y_v\) is the value allocated to player \(v\). An allocation \(y\) is \emph{stable} if no subset of players has an incentive to deviate from the current set of coalitions, to form their own coalition. Mathematically, \(y\) is stable if
\begin{equation*}
    \sum_{v \in S} y_v \geq \nu(G[S]) \quad \forall S \subseteq V.
\end{equation*}
The \emph{core} of an assignment game is the set of all stable allocations when the \emph{grand coalition} (\(S = V\)) is formed:
\begin{equation*}
    \core(G) = \left\{y \in \mathbb{R}^V_{\geq0}: 
        \sum_{v \in S} y_v \geq \nu(G[S]) \ \forall S \subseteq V, 
        \sum_{v \in V} y_v = \nu(G) \right\}.
\end{equation*}

In this paper, we study a \emph{two-stage stochastic version} of the assignment game. Studying combinatorial problems in the two-stage setting is a popular area of research (see e.g., \cite{Bampis2021Online,Charikar2005Sampling,Gupta2014Changing,Lee2020Maximum,Ravi2006Hedging,Swamy2008Sample,Swamy2012Sampling-based}). Recently, this setting has been studied for prominent game theory problems such as stable matchings~\cite{Bampis2023Online,Faenza2024Two-stage}.
In particular, the authors of~\cite{Faenza2024Two-stage} study a two-stage stochastic stable matching problem where in the second stage the set of vertices changes. We here analyze this setting in the context of the assignment game.

Given a first-stage assignment game instance, in a second stage we can have some players leaving the game, new players joining the game, and/or some additions and removals in the edge set of the original instance.
Formally, we represent this as having a new graph describing the instance in the second stage, that can be any bipartite graph as long as it keeps the same bipartition as the first-stage graph for the vertices that stay in the game. 
The second-stage instance is sampled from some distribution \(\D\). We denote the starting (first-stage) graph by \(G_0 = (V_0, E_0)\), and for any second-stage scenario \(S \sim \D\), we denote the corresponding graph by \(G_S = (V_S, E_S)\). 
The goal is to minimize the expected total loss of the remaining players (i.e.\ decrease in allocation value). Mathematically, the two-stage stochastic assignment game is
\begin{equation}
    \label{eq: 2SAG}
    \tag{2SAG}
    \min_{y \in \core(G_0)} \E_{S \sim \D} \bra{ \min_{y^S \in \core(G_S)} \sum_{v \in V_0 \cap V_S} \lambda_v \bra{y_v - y^S_v}^+ },
\end{equation}
where \([x]^+ = \max\cbra{0, x}\), and \(\lambda \geq 0\) is the dissatisfaction of players per unit loss of allocation value.

\paragraph{Our results.}
    We first consider the setting where the probability distribution \(\D\) is given explicitly in \Cref{sec:explicit}. We observe that the problem can be modeled as a linear program (LP), and hence it is solvable in polynomial time. Interestingly, we prove that the feasible region is an \emph{integral} polyhedron. For this, we show that the problem can be modeled as a flow problem in a suitable auxiliary graph, and then exploit duality. We leverage this integrality result in two ways.

    First, we exploit it when considering a probability distribution given implicitly, as described in \Cref{sec:implicit}.
    The integrality result allows us to mimic the arguments used in \cite{Faenza2024Two-stage} for two-stage stable matching,
    hence showing that in this setting the problem is computationally hard to solve, but it can be approximated using the well-known sample average approximation (SAA) method~\cite{Kleywegt2002Sample}.

    Second, the integrality property reveals a close relationship with the well-known \emph{multistage vertex cover} problem, which we discuss in \Cref{sec:multistage}. It is known that the multistage vertex cover problem is NP-hard even with only two stages and bipartite graphs at each stage \cite{Fluschnik2022Mutlistage}. However, as a consequence of our findings, we can show that the problem becomes polynomial-time solvable when the bipartition remains consistent across all stages. 
\subsection{Preliminaries}

The standard linear programming relaxation of \(\nu(G)\) is defined as
\begin{equation*}
    \nu_f(G) = \max\cbra{\sum_{e \in E} x_e : \sum_{u: uv \in E} x_{uv} \leq 1 \ \forall v \in V, x \geq 0}.
\end{equation*}
The dual of \(\nu_f(G)\) is
\begin{equation*}
    \tau_f(G) = \min\cbra{1^\top y : y_u + y_v \geq 1 \ \forall uv \in E, y \geq 0},
\end{equation*}
where we denote by \(1^\top y = \sum_{v \in V} y_v\).
We refer to the \(y\) in this formulation as a fractional vertex cover, and we let \(\tau(G)\) be the value of a minimum (integral) vertex cover. It follows from Königs theorem (\(\nu(G) = \tau(G)\) for bipartite graphs) and LP theory (\(\nu(G) \leq \nu_f(G) = \tau_f(G) \leq \tau(G)\)) that \(\tau_f(G) = \nu(G)\) for bipartite graphs.
Shapley and Shubik~\cite{Shapley1971Assignment} showed that the core of an assignment game is precisely the set of minimum fractional vertex covers, i.e.,
\begin{equation*}
    \core(G) = \left\{y \in \mathbb{R}^V_{\geq0}: 
        y_u + y_v \geq 1 \ \forall uv \in E, 
        1^\top y = \nu(G) \right\}.
\end{equation*}
From this and \(\tau_f(G) = \nu(G)\) it readily follows that the core of each assignment game is nonempty: there is always a minimum fractional vertex cover. This is important for our two-stage stochastic assignment game, because we are assured that in both stages the core is nonempty.

Observe that in any core element \(y \leq 1\), because \(y_u + y_v \geq 1\) for all edges in a maximum matching \(M\), \(1^\top y = \nu(G) = |M|\), and \(y \geq 0\).
\section{Explicit Distribution}
\label{sec:explicit}

Suppose the distribution \(\D\) is given explicitly by a list of scenarios \(\S\) and their respective probabilities of occurrence \(\cbra{p_S}_{S \in \S}\). 
Here we will consider the problem of minimizing the absolute difference, instead of the positive difference, i.e., \(|y_v - y^S_v|\) instead of \([y_v - y^S_v]^+\). We solve this in such a way that one can later choose either option, or even \([y^S_v - y_v]^+\).
Using the scenarios \(\S\), we can expand the expectation in \eqref{eq: 2SAG}:
\begin{equation}
    \label{eq: 2SAG explicit}
    \tag{2SAG-expl}
    \min_{y \in \core(G_0)} \sum_{S \in \S} p_S \bra{ \min_{y^S \in \core(G_S)} \sum_{v \in V_0 \cap V_S} \lambda_v \Abs{y_v - y^S_v} }.
\end{equation}
We can rewrite this as the following LP.
\begin{equation}
\label{eq: 2SAG LP}
\tag{2SAG-LP}
\begin{split}
    \min \quad 
        & \sum_{S \in \S} p_S \sum_{v \in V_0 \cap V_S} \lambda_v (\delta^S_v + d^S_v) \\
    \text{s.t.} \quad
        & y_u + y_v \geq 1 \quad \forall uv \in E_0 \\
        & 1^\top y = \nu(G_0) \\
        & y \in \R^{V_0}_{\geq 0} \\
        & y^S_u + y^S_v \geq 1 \quad \forall uv \in E_S, \forall S \in \S \\
        & 1^\top y^S = \nu(G_S) \quad \forall S \in \S \\
        & y^S \in \R^{V_S}_{\geq 0} \quad \forall S \in \S \\
        & y_v -  y^S_v \leq \delta^S_v \quad \forall v \in V_0 \cap V_S, \forall S \in \S \\
        & y^S_v -  y_v \leq d^S_v \quad \forall v \in V_0 \cap V_S, \forall S \in \S \\
        & \delta^S \in \R^{V_0 \cap V_S}_{\geq0} \quad \forall S \in \S \\
        & d^S \in \R^{V_0 \cap V_S}_{\geq0} \quad \forall S \in \S
\end{split}
\end{equation}
Observe that if we change \(\delta^S_v + d^S_v\) in the objective to \(\delta^S_v\) or \(d^S_v\), then the objective we consider is \([y_v - y^S_v]^+\) or \([y^S_v - y_v]^+\), respectively.

Let \(V = V_0 \cup \bigcup_{S \in \S} V_S\). This LP has \(O(|V| |\S|)\) variables and \(O(|V|^2 |\S|)\) constraints. So it has size polynomial in the input size, which means we can solve \eqref{eq: 2SAG explicit} in polynomial-time, by solving \eqref{eq: 2SAG LP}. Using an auxiliary linear program we show that the feasible region of \eqref{eq: 2SAG LP} is an integral polyhedron.

\begin{theorem}
    \label{thm: integral solution}
    The feasible region of \eqref{eq: 2SAG LP} is an integral polyhedron.
\end{theorem}

We will consider \eqref{eq: 2SAG LP} with an arbitrary objective, and show that for each objective we can find an integral optimal solution. Since for each extreme point of a polyhedron there is an objective such that the extreme point is the unique optimal solution, this proves that the polyhedron is integral. We consider the following objective, where \(\alpha\) can take any value, \(\beta \geq 0\) and \(b \geq 0\). Note that if \(\beta^S_v < 0\) or \(b^S_v < 0\) for any \(v \in V_0 \cap V_S\), \(S \in \S\), then the LP becomes unbounded, and so there are no extreme points in those directions.
\begin{equation*}
    \min \quad 
        \sum_{v \in V_0} \alpha_v y_v + \sum_{S \in \S} \sum_{v \in V_S} \alpha^S_v y^S_v+ \sum_{S \in \S} \sum_{v \in V_0 \cap V_S} \beta^S_v \delta^S_v + b^S_v d^S_v
\end{equation*}

We will formulate a maximum flow problem on an auxiliary graph, and compute its dual LP. It is known that the constraint matrix of a maximum flow LP is always totally unimodular (TU). The dual LP has the transpose as constraint matrix, which is then also TU. Since in addition the right hand side of the dual constraints are integral, this means the feasible region of the dual is an integral polyhedron. (See e.g.~\cite{Schrijver2003Combinatorial} for properties of TU matrices.) Finally, we show that we can map an optimal dual solution to an optimal solution for \eqref{eq: 2SAG LP}.

Let \(V_1, V_2 \subseteq V_0 \cup \bigcup_{S \in \S} V_S\) be the bipartition of \(G_0\) merged with all \(G_S\). Let
\begin{equation*}
    \varepsilon = \frac{1}{1 + \sum_{v \in V_0} |\alpha_v| + \sum_{S \in \S} \sum_{v \in V_S} |\alpha^S_v| + \sum_{S \in \S} \sum_{v \in V_0 \cap V_S} \beta^S_v + b^S_v}.
\end{equation*}
Note that \(\varepsilon > 0\).
We create the auxiliary graph \(G' = (V', A)\) as follows. Let
\begin{equation*}
    V' = \cbra{s, t} \cup V_0 \cup \cbra{v^S : v \in V_S, S \in \S}.
\end{equation*}
The arc set \(A\) contains the following arcs:
\begin{itemize}
    \item \(sv\) for all \(v \in V_0 \cap V_1\), with flow-capacity \(1 + \varepsilon \alpha_v\);
    \item \(sv^S\) for all \(v \in V_S \cap V_1\), \(S \in \S\), with flow-capacity \(1 + \varepsilon \alpha^S_v\);
    \item \(v^S v\) for all \(v \in (V_0 \cap V_S) \cap V_1\), \(S \in \S\), with flow-capacity \(\varepsilon \beta^S_v\);
    \item \(v v^S\) for all \(v \in (V_0 \cap V_S) \cap V_1\), \(S \in \S\), with flow-capacity \(\varepsilon b^S_v\);
    \item \(vt\) for all \(v \in V_0 \cap V_2\), with flow-capacity \(1 + \varepsilon \alpha_v\);
    \item \(v^S t\) for all \(v \in V_S \cap V_2\), \(S \in \S\), with flow-capacity \(1 + \varepsilon \alpha^S_v\);
    \item \(v v^S\) for all \(v \in (V_0 \cap V_S) \cap V_2\), \(S \in \S\), with flow-capacity \(\varepsilon \beta^S_v\);
    \item \(v^S v\) for all \(v \in (V_0 \cap V_S) \cap V_2\), \(S \in \S\), with flow-capacity \(\varepsilon b^S_v\);
    \item \(uv\) for all \(uv \in E_0\), such that \(u \in V_1\), \(v \in V_2\), without upperbound on the flow-capacity;
    \item \(u^S v^S\) for all \(uv \in E_S\), \(S \in \S\), such that \(u \in V_1\), \(v \in V_2\), without upperbound on the flow-capacity.
\end{itemize}
\Cref{fig:aux graph} shows an example of the auxiliary graph in general, and \Cref{fig:aux graph explicit} shows an example of the auxiliary graph of a specific instance.
\begin{figure}[ht]
    \centering
    \begin{subfigure}[t]{.4\textwidth}
        \centering
        \begin{tikzpicture}[
            scale=.6,
            circ/.style={
                circle,
                draw=black,
                inner sep=2pt
            },
        ]
            \node[circ,label={above:\(s\)}] (s) at (0,0) {};
            \node[circ,label={above:\(v^S\)}] (a) at (2,1) {};
            \node[circ,label={above:\(v\)}] (v) at (4,0) {};         
        
            \draw[->] (s) to (v);
            \draw[->] (s) to (a);
            \draw[->,bend left=20] (a) to (v);
            \draw[->,bend left=20] (v) to (a);
            \draw[->,dashed] (a) to (4,2);
            \draw[->,dashed] (v) to (6,0);
        \end{tikzpicture}
        \caption{Example of the source/\(V_1\) side.}
        \label{subfig: aux graph source}
    \end{subfigure}
    \hspace{.7cm}
    \begin{subfigure}[t]{.4\textwidth}
        \centering
        \begin{tikzpicture}[
            scale=.6,
            circ/.style={
                circle,
                draw=black,
                inner sep=2pt
            },
        ]
            \node[circ,label={above:\(v\)}] (v) at (0,0) {};
            \node[circ,label={above:\(v^S\)}] (a) at (2,1) {};
            \node[circ,label={above:\(t\)}] (t) at (4,0) {};         
        
            \draw[->] (v) to (t);
            \draw[->,bend left=20] (v) to (a);
            \draw[->,bend left=20] (a) to (v);
            \draw[->] (a) to (t);
            \draw[->,dashed] (-2,0) to (v);
            \draw[->,dashed] (0,2) to (a);
        \end{tikzpicture}
        \caption{Example of the sink/\(V_2\) side.}
        \label{subfig: aux graph sink}
    \end{subfigure}
    \caption{Example of part of the auxiliary graph, where dashed arcs indicate the edges corresponding with \(E_0\) and \(E_S\).}
    \label{fig:aux graph}
\end{figure}
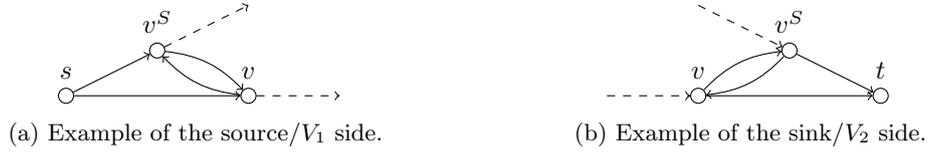
\begin{figure}[ht]
    \centering
    \begin{tikzpicture}[
        scale=.6,
        circ/.style={
            circle,
            draw=black,
            inner sep=2pt
        },
    ]
        \node[circ,label={left:\(s\)}] (s) at (0,0) {};
        \node[circ,label={above:\(a^{S_1}\)}] (a1) at (2,2) {};
        \node[circ,label={below:\(a^{S_2}\)}] (a2) at (2,-2) {};
        \node[circ,label={below:\(a\)}] (a) at (4,0) {};
        \node[circ,label={below:\(b\)}] (b) at (6,1) {};
        \node[circ,label={below:\(c\)}] (c) at (6,-1) {};
        \node[circ,label={above:\(b^{S_1}\)}] (b1) at (8,2) {};
        \node[circ,label={below:\(c^{S_2}\)}] (c1) at (8,-2) {};
        \node[circ,label={right:\(t\)}] (t) at (10,0) {};
    
        \draw[->] (s) to (a);
        \draw[->] (s) to (a1);
        \draw[->] (s) to (a2);
        \draw[->,bend left=20] (a1) to (a);
        \draw[->,bend left=20] (a) to (a1);
        \draw[->,bend left=20] (a2) to (a);
        \draw[->,bend left=20] (a) to (a2);
        \draw[->] (a) to (b);
        \draw[->] (a) to (c);
        \draw[->,bend left=20] (b) to (b1);
        \draw[->,bend left=20] (b1) to (b);
        \draw[->,bend left=20] (c) to (c1);
        \draw[->,bend left=20] (c1) to (c);
        \draw[->] (b) to (t);
        \draw[->] (c) to (t);
        \draw[->] (b1) to (t);
        \draw[->] (c1) to (t);
        \draw[->] (a1) to (b1);
        \draw[->] (a2) to (c1);
    \end{tikzpicture}
    \caption{Example of the auxiliary graph for the instance given by \(V_0 = \{a,b,c\}\), \(E_0 = \{ab, ac\}\), and \(\S=\{S_1, S_2\}\) where \(V_{S_1} = \{a, b\}\), \(E_{S_1} = \{ab\}\), and \(V_{S_2} = \{a, c\}\), \(E_{S_2} = \{ac\}\).}
    \label{fig:aux graph explicit}
\end{figure}
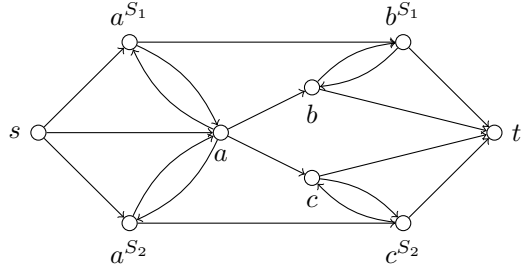

Let \(\1[...]\) be \(1\) if the statement in between the brackets is true, and \(0\) if the statement is false.
We can formulate the maximum flow problem in \(G'\) as an LP as follows.
\begin{equation}
\label{eq: flow LP}
\begin{split}
    \max \quad 
        & \sum_{v \in V_0 \cap V_1} f_{sv} + \sum_{S \in \S} \sum_{v \in V_S \cap V_1} f_{sv^S} \\
    \text{s.t.} \quad
        & f_{sv} + \sum_{S \in \S: v \in V_S} (f_{v^Sv} - f_{v v^S}) - \sum_{u: uv \in E_0} f_{vu} = 0 \quad \forall v \in V_0 \cap V_1 \\
        & f_{sv^S} + \1[v \in V_0](f_{v v^S} - f_{v^Sv}) - \sum_{u: uv \in E_S} f_{v^Su^S} = 0 \quad \forall v \in V_S \cap V_1, \forall S \in \S \\
        & \sum_{u: uv \in E_0} f_{uv} + \sum_{S \in \S: v \in V_S} (f_{v^S v} - f_{vv^S}) - f_{vt} = 0  \quad \forall v \in V_0 \cap V_2 \\
        & \sum_{u: uv \in E_S} f_{u^Sv^S} + \1[v \in V_0](f_{vv^S} - f_{v^S v}) - f_{v^St}= 0 \quad \forall v \in V_S \cap V_2, \forall S \in \S \\
        & f_{sv} \leq 1 + \varepsilon \alpha_v \quad \forall v \in V_0 \cap V_1 \\
        & f_{sv^S} \leq 1 + \varepsilon \alpha^S_v \quad \forall v \in V_S \cap V_1, \forall S \in \S \\
        & f_{vt} \leq 1 + \varepsilon \alpha_v \quad \forall v \in V_0 \cap V_2 \\
        & f_{v^St} \leq 1 + \varepsilon \alpha^S_v \quad \forall  v \in V_S \cap V_2, \forall S \in \S \\
        & f_{v^Sv} \leq \varepsilon \beta^S_v \quad \forall v \in (V_0 \cap V_S) \cap V_1, \forall S \in \S \\
        & f_{v v^S} \leq \varepsilon b^S_v \quad \forall v \in (V_0 \cap V_S) \cap V_1, \forall S \in \S \\
        & f_{vv^S} \leq \varepsilon \beta^S_v \quad \forall v \in (V_0 \cap V_S) \cap V_2, \forall S \in \S \\
        & f_{v^S v} \leq \varepsilon b^S_v \quad \forall v \in (V_0 \cap V_S) \cap V_2, \forall S \in \S \\
        & f \in \R^A_{\geq 0}
\end{split}
\end{equation}
For this LP to have a feasible solution, the flow-capacities need to be nonnegative: \(1 + \varepsilon \alpha_v \geq 0\) for all \(v \in V_0\), \(1 + \varepsilon \alpha^S_v \geq 0\) for all \(v \in V_S\), \(S \in \S\), \(\varepsilon \beta^S_v \geq 0\) for all \(v \in V_0 \cap V_S\), \(S \in \S\), and \(\varepsilon b^S_v \geq 0\) for all \(v \in V_0 \cap V_S\), \(S \in \S\). The latter two are satisfied as \(\varepsilon > 0\) and we set \(\beta, b \geq 0\). For any \(v \in V_0\), we have
\begin{equation*}
    \alpha_v 
    \geq - |\alpha_v|
    \geq - \left( 1 + \sum_{v \in V_0} |\alpha_v| + \sum_{S \in \S} \sum_{v \in V_S} |\alpha^S_v| + \sum_{S \in \S} \sum_{v \in V_0 \cap V_S} \beta^S_v + b^S_v \right)
    = -\frac{1}{\varepsilon},
\end{equation*}
and so, since \(\varepsilon > 0\), we have \(1 + \varepsilon \alpha_v \geq 1 + \varepsilon \frac{-1}{\varepsilon} = 0\). The same argument works for \(\alpha^S_v\) for any \(v \in V_S\), \(S \in \S\). 

The dual of this flow LP is as follows.
\begin{equation}
\label{eq: flow LP dual}
\begin{split}
    \min \quad 
        & \sum_{v \in V_0} y_v + \sum_{S \in \S} \sum_{v \in V_S} y^S_v + \varepsilon \left( \sum_{v \in V_0} \alpha_v y_v + \sum_{S \in \S} \sum_{v \in V_S} \alpha^S_v y^S_v + \sum_{S \in \S} \sum_{v \in V_0 \cap V_S} \beta^S_v \delta^S_v + b^S_v d^S_v \right)\\
    \text{s.t.} \quad
        & \gamma_v + y_v \geq 1 \quad \forall v \in V_0 \cap V_1 \\
        & \gamma^S_v + y^S_v \geq 1 \quad \forall v \in V_S \cap V_1, \forall S \in \S \\
        & \gamma_v - \gamma^S_v + \delta^S_v \geq 0 \quad \forall v \in (V_0 \cap V_S) \cap V_1, \forall S \in \S \\
        & \gamma^S_v - \gamma_v + d^S_v \geq 0 \quad \forall v \in (V_0 \cap V_S) \cap V_1, \forall S \in \S \\
        & \gamma_v - \gamma_u \geq 0  \quad \forall uv \in E_0 \text{ such that } u \in V_1, v \in V_2 \\
        & \gamma^S_v - \gamma^S_u \geq 0 \quad \forall uv \in E_S \text{ such that } u \in V_1, v \in V_2, \forall S \in \S \\
        & - \gamma_v + y_v \geq 0 \quad \forall v \in V_0 \cap V_2 \\
        & - \gamma^S_v + y^S_v \geq 0 \quad \forall v \in V_S \cap V_2, \forall S \in \S \\
        & \gamma^S_v - \gamma_v + \delta^S_v \geq 0 \quad \forall v \in (V_0 \cap V_S) \cap V_2, \forall S \in \S \\
        & \gamma_v - \gamma^S_v + d^S_v \geq 0 \quad \forall v \in (V_0 \cap V_S) \cap V_2, \forall S \in \S \\
        & y \in \R^{V_0}_{\geq 0} \\
        & y^S \in \R^{V_S}_{\geq 0} \quad \forall S \in \S \\
        & \delta^S \in \R^{V_0 \cap V_S}_{\geq 0} \quad \forall S \in \S \\
        & d^S \in \R^{V_0 \cap V_S}_{\geq 0} \quad \forall S \in \S \\
        & \gamma \in \R^{V_0} \\
        & \gamma^S \in \R^{V_S} \quad \forall S \in \S
\end{split}
\end{equation}

\begin{lemma}
    \label{lem: integral polyhedron}
    The feasible region of \eqref{eq: flow LP dual} is an integral polyhedron.
\end{lemma}
\begin{proof}
    The constraint matrix of a maximum flow LP is always TU; in particular the constraint matrix of \eqref{eq: flow LP} is TU.

    The transpose of a TU matrix is TU. The constraint matrix of a dual LP is the transpose of the constraint matrix of the primal LP. Together they imply that the constraint matrix of \eqref{eq: flow LP dual} is TU.

    Finally, a polyhedron with a TU constraint matrix and integral right hand side is an integral polyhedron; in particular, the feasible region of \eqref{eq: flow LP dual} is an integral polyhedron.
\end{proof}

We can obtain an integral optimal solution for \eqref{eq: flow LP dual} by solving \eqref{eq: flow LP dual} directly, or combinatorially, as follows. First we find an optimal flow in the auxiliary graph. The flow can then be used to obtain an optimal solution for \eqref{eq: flow LP dual}, by using complementary slackness. Finally, if this is not an integral solution, then in particular it is not an extreme point solution. So we can use this solution to go to an extreme point solution, which will be integral by \Cref{lem: integral polyhedron}.

We map an integral optimal solution for \eqref{eq: flow LP dual} to an integral optimal solution for \eqref{eq: 2SAG LP} in two steps. First we map it to the following LP.
\begin{equation}
\label{eq: 2SAG LP eps}
\begin{split}
\min \quad 
    & 1^\top y + \sum_{S \in \S} 1^\top y^S + \varepsilon \left( \sum_{v \in V_0} \alpha_v y_v + \sum_{S \in \S} \sum_{v \in V_S} \alpha^S_v y^S_v + \sum_{S \in \S} \sum_{v \in V_0 \cap V_S} \beta^S_v \delta^S_v + b^S_v d^S_v \right) \\
\text{s.t.} \quad
    & y_u + y_v \geq 1 \quad \forall uv \in E_0 \\
    & y \in \R^{V_0}_{\geq 0} \\
    & y^S_u + y^S_v \geq 1 \quad \forall uv \in E_S, \forall S \in \S \\
    & y^S \in \R^{V_S}_{\geq 0} \quad \forall S \in \S \\
    & y_v -  y^S_v \leq \delta^S_v \quad \forall v \in V_0 \cap V_S, \forall S \in \S \\
    & y^S_v -  y_v \leq d^S_v \quad \forall v \in V_0 \cap V_S, \forall S \in \S \\
    & \delta^S \in \R^{V_0 \cap V_S}_{\geq 0} \quad \forall S \in \S \\
    & d^S \in \R^{V_0 \cap V_S}_{\geq 0} \quad \forall S \in \S
\end{split}
\end{equation}

\begin{lemma}
    \label{lem: flow dual to 2SAG eps}
    We can obtain an optimal solution for \eqref{eq: 2SAG LP eps} from an optimal solution for \eqref{eq: flow LP dual}.
\end{lemma}
\begin{proof}
    \begin{claim}
        Any feasible solution for \eqref{eq: 2SAG LP eps} can be mapped to a feasible solution for \eqref{eq: flow LP dual} with the same objective value.
    \end{claim}
    \begin{proof}
        Let \(\hat{\mu} = (\hat{y}, \hat{y}^S \text{ for } S \in S, \hat{\delta}^S \text{ for } S \in S, \hat{d}^S \text{ for } S \in S)\) be a feasible solution for \eqref{eq: 2SAG LP eps}. We extend it to a solution \(\hat{\mu}_{\text{ext}}\) for \eqref{eq: flow LP dual} by setting \(\hat{\gamma}_v = 1 - \hat{y}_v\) for all \(v \in V_0 \cap V_1\), \(\hat{\gamma}_v = \hat{y}_v\) for all \(v \in V_0 \cap V_2\), \(\hat{\gamma}^S_v = 1 - \hat{y}^S_v\) for all \(v \in V_S \cap V_1\), \(S \in \S\), and \(\hat{\gamma}^S_v = \hat{y}^S_v\) for all \(v \in V_S \cap V_2\), \(S \in \S\). 
        It is clear that \(\hat{\mu}_{\text{ext}}\) has the same objective value as \(\hat{\mu}\), as the objective functions of the two linear programs are the same and do not involve the \(\gamma\) variables. We next show that \(\hat{\mu}_{\text{ext}}\) is feasible.

        The constraints \(\gamma_v + y_v \geq 1\), \(\gamma^S_v + y^S_v \geq 1\), \(- \gamma_v + y_v \geq 0\) and \(- \gamma^S_v + y^S_v \geq 0\) of \eqref{eq: flow LP dual} are satisfied by \(\hat{\mu}_{\text{ext}}\) by definition. Nonnegativity of \(y\), \(y^S\), \(\delta^S\) and \(d^S\) for all \(S \in \S\) are satisfied by \(\hat{\mu}_{\text{ext}}\), because \(\hat{\mu}\) is feasible for \eqref{eq: 2SAG LP eps}.

        Let \(S \in \S\) and \(v \in (V_0 \cap V_S) \cap V_1\). We have
        \begin{equation*}
            \hat{\gamma}_v - \hat{\gamma}^S_v + \hat{\delta}^S_v
            = 1 - \hat{y}_v - (1 - \hat{y}^S_v) + \hat{\delta}^S_v
            = \hat{y}^S_v - \hat{y}_v + \hat{\delta}^S_v
            \geq 0,
        \end{equation*}
        and
        \begin{equation*}
            \hat{\gamma}^S_v - \hat{\gamma}_v + \hat{d}^S_v
            = 1 - \hat{y}^S_v - (1 - \hat{y}_v) + \hat{d}^S_v
            = \hat{y}_v - \hat{y}^S_v + \hat{d}^S_v
            \geq 0,
        \end{equation*}
        where both times the last inequality follows from \(\hat{\mu}\)'s feasibility for \eqref{eq: 2SAG LP eps}. Similarly, we have for \(S \in \S\) and \(v \in (V_0 \cap V_S) \cap V_2\)
        \begin{equation*}
            \hat{\gamma}^S_v - \hat{\gamma}_v + \hat{\delta}^S_v
            = \hat{y}^S_v - \hat{y}_v + \hat{\delta}^S_v
            \geq 0,
        \end{equation*}
        and
        \begin{equation*}
            \hat{\gamma}_v - \hat{\gamma}^S_v + \hat{d}^S_v
            = \hat{y}_v - \hat{y}^S_v + \hat{d}^S_v
            \geq 0.
        \end{equation*}
        Let \(uv \in E_0\) such that \(u \in V_1\) and \(v \in V_2\). We have
        \begin{equation*}
            \hat{\gamma}_v - \hat{\gamma}_u
            = \hat{y}_v - (1 - \hat{y}_u)
            = \hat{y}_v + \hat{y}_u - 1
            \geq 0.
        \end{equation*}
        Finally, let \(S \in \S\) and \(uv \in E_S\) such that \(u \in V_1\) and \(v \in V_2\). We have
        \begin{equation*}
            \hat{\gamma}^S_v - \hat{\gamma}^S_u
            = \hat{y}^S_v - (1 - \hat{y}^S_u)
            = \hat{y}^S_v + \hat{y}^S_u - 1
            \geq 0.
        \end{equation*}
    \end{proof}

    \begin{claim}
        Any feasible solution to \eqref{eq: flow LP dual} can be mapped to a feasible solution to \eqref{eq: 2SAG LP eps} with the same objective value.
    \end{claim}
    \begin{proof}
        Let \(\hat{\mu} = (\hat{y}, \hat{y}^S \text{ for } S \in S, \hat{\delta}^S \text{ for } S \in S, \hat{d}^S \text{ for } S \in S, \hat{\gamma}, \hat{\gamma}^S \text{ for } S \in S)\) be a feasible solution for \eqref{eq: flow LP dual}. We restrict it to a solution \(\hat{\mu}_{\text{res}}\) for \eqref{eq: 2SAG LP eps} by disregarding the \(\gamma\) variables. It is clear that \(\hat{\mu}_{\text{res}}\) has the same objective value as \(\hat{\mu}\), as the objective functions of the two linear programs are the same and do not involve the \(\gamma\) variables. We next show that \(\hat{\mu}_{\text{res}}\) is feasible.

        First observe that \(\hat{\mu}_{\text{res}} \geq 0\).

        Let \(uv \in E_0\) such that \(u \in V_1\) and \(v \in V_2\). We have
        \begin{equation*}
            \hat{y}_u + \hat{y}_v
            = (\hat{\gamma}_u + \hat{y}_u) + (\hat{\gamma}_v - \hat{\gamma}_u) + (- \hat{\gamma}_v + \hat{y}_v)
            \geq 1 + 0 + 0 = 1.
        \end{equation*}
        Similarly, let \(S \in \S\) and \(uv \in E_S\) such that \(u \in V_1\) and \(v \in V_2\). We have
        \begin{equation*}
            \hat{y}^S_u + \hat{y}^S_v
            = (\hat{\gamma}^S_u + \hat{y}^S_u) + (\hat{\gamma}^S_v - \hat{\gamma}^S_u) + (- \hat{\gamma}^S_v + \hat{y}^S_v)
            \geq 1 + 0 + 0 = 1.
        \end{equation*}

        Let \(S \in \S\), \(v \in (V_0 \cap V_S) \cap V_1\). We will show that without loss of generality we can assume that \(\hat{\gamma}_v + \hat{y}_v = 1\) and \(\hat{\gamma}^S_v + \hat{y}^S_v = 1\). Consequently, we have
        \begin{equation*}
        \begin{split}
            \hat{\delta}^S_v + \hat{y}^S_v - \hat{y}_v
            &= \hat{\delta}^S_v + \hat{y}^S_v - \hat{y}_v + \hat{\gamma}^S_v - \hat{\gamma}^S_v + \hat{\gamma}_v - \hat{\gamma}_v \\
            &= (\hat{\delta}^S_v + \hat{\gamma}_v - \hat{\gamma}^S_v) + (\hat{\gamma}^S_v + \hat{y}^S_v) - (\hat{\gamma}_v + \hat{y}_v) \\
            & \geq 0 + 1 - 1 = 0,
        \end{split}
        \end{equation*}
        and
        \begin{equation*}
        \begin{split}
            \hat{d}^S_v + \hat{y}_v - \hat{y}^S_v
            &= \hat{d}^S_v + \hat{y}_v - \hat{y}^S_v + \hat{\gamma}_v - \hat{\gamma}_v + \hat{\gamma}^S_v - \hat{\gamma}^S_v \\
            &= (\hat{d}^S_v + \hat{\gamma}^S_v - \hat{\gamma}_v) + (\hat{\gamma}_v + \hat{y}_v) - (\hat{\gamma}^S_v + \hat{y}^S_v) \\
            & \geq 0 + 1 - 1 = 0.
        \end{split}
        \end{equation*}
        Now, let \(S \in \S\), \(v \in (V_0 \cap V_S) \cap V_2\). We will also show that without loss of generality we can assume that \(- \hat{\gamma}_v + \hat{y}_v = 0\) and \(- \hat{\gamma}^S_v + \hat{y}^S_v = 0\). Consequently, we have
        \begin{equation*}
        \begin{split}
            \hat{\delta}^S_v + \hat{y}^S_v - \hat{y}_v
            &= \hat{\delta}^S_v + \hat{y}^S_v - \hat{y}_v + \hat{\gamma}^S_v - \hat{\gamma}^S_v + \hat{\gamma}_v - \hat{\gamma}_v \\
            &= (\hat{\delta}^S_v + \hat{\gamma}^S_v - \hat{\gamma}_v) + (- \hat{\gamma}^S_v + \hat{y}^S_v) - (- \hat{\gamma}_v + \hat{y}_v) \\
            & \geq 0 + 0 - 0 = 0,
        \end{split}
        \end{equation*}
        and
        \begin{equation*}
        \begin{split}
            \hat{d}^S_v + \hat{y}_v - \hat{y}^S_v
            &= \hat{d}^S_v + \hat{y}_v - \hat{y}^S_v + \hat{\gamma}_v - \hat{\gamma}_v + \hat{\gamma}^S_v - \hat{\gamma}^S_v \\
            &= (\hat{d}^S_v + \hat{\gamma}_v - \hat{\gamma}^S_v) + (- \hat{\gamma}_v + \hat{y}_v) - (- \hat{\gamma}^S_v + \hat{y}^S_v) \\
            & \geq 0 + 0 - 0 = 0,
        \end{split}
        \end{equation*}
        This finishes the feasibility proof of \(\hat{\mu}_{\text{res}}\).

        To show that we can indeed assume without loss of generality that \(\hat{\gamma}_v + \hat{y}_v = 1\) for \(v \in V_0 \cap V_1\), suppose it does not hold, so: \(\hat{\gamma}_v + \hat{y}_v > 1\). If \(\hat{y}_v > 0\), then lower \(\hat{y}_v\) by \(\min\{\hat{y}_v, 1 - \hat{\gamma}_v\}\) (this does not affect feasibility as \(\hat{y}_v\) is contained in only this one constraint, and it improves the objective). If now \(\hat{\gamma}_v + \hat{y}_v = 1\), then we are done. If not, then it must be that \(\hat{y}_v = 0\) and \(\hat{\gamma}_v > 1\). Let \(0 < \eta \leq \hat{\gamma}_v - 1\) and lower \(\hat{\gamma}_v\) by \(\eta\): \(\hat{\gamma}'_v = \hat{\gamma}_v - \eta\) (this does not change the objective). By definition of \(\eta\), we still have \(\hat{\gamma}'_v + \hat{y}_v \geq 1\). Let \(u \in V_2\) such that \(uv \in E_0\), then we have 
        \begin{equation*}
            \hat{\gamma}_u - \hat{\gamma}'_v 
            = \hat{\gamma}_u - (\hat{\gamma}_v - \eta)
            = \hat{\gamma}_u - \hat{\gamma}_v + \eta
            \geq \eta 
            > 0.
        \end{equation*}
        If \(v \notin V_S\) for all \(S \in \S\), then the solution with \(\hat{\gamma}_v\) replaced by \(\hat{\gamma}'_v\) is feasible. If there is at least one \(S \in \S\) such that \(v \in V_S\), then we have
        \begin{equation*}
            \hat{\gamma}^S_v - \hat{\gamma}'_v + \hat{d}^S_v
            = \hat{\gamma}^S_v - (\hat{\gamma}_v - \eta) + \hat{d}^S_v
            \geq \eta
            > 0,
        \end{equation*}
        for all \(S \in \S\) such that \(v \in V_S\).
        If also \(\hat{\gamma}'_v - \hat{\gamma}^S_v + \hat{\delta}^S_v \geq 0\) for some choice of \(\eta\), then the solution with \(\hat{\gamma}_v\) replaced by \(\hat{\gamma}'_v\) is again feasible.

        So now suppose that there is some \(S \in \S\) with \(v \in V_S\), such that for all choices of \(\eta\), this latter constraint is not satisfied. Then it must be that \(\hat{\gamma}_v - \hat{\gamma}^S_v + \hat{\delta}^S_v = 0\). We can make this constraint work if we also decrease \(\hat{\gamma}^S_v\) by \(\eta\): \((\hat{\gamma}^S_v)' = \hat{\gamma}^S_v - \eta\) (this does not change the objective). Like before, it is clear that the constraint \(\gamma^S_u - \gamma^S_v \geq 0\) is still satisfied, as we increase the left hand side.
        The constraint \(\gamma^S_v - \gamma_v + d^S_v \geq 0\) is also still satisfied, as we decrease \(\hat{\gamma}^S_v\) and \(\hat{\gamma}_v\) by the same amount.
        Finally, using \(\hat{\gamma}_v - \hat{\gamma}^S_v + \hat{\delta}^S_v = 0\) and \(\hat{\gamma}_v - \eta \geq 1\), we have
        \begin{equation*}
            (\hat{\gamma}^S_v)' + \hat{y}^S_v
            = \hat{\gamma}^S_v - \eta + \hat{y}^S_v
            = \hat{\gamma}_v + \hat{\delta}^S_v - \eta + \hat{y}^S_v
            \geq 1 + 0 + 0
            = 1.
        \end{equation*}
        So again, we find that the solution with \(\hat{\gamma}_v\) replaced by \(\hat{\gamma}'_v\) is feasible. Consequently, we can decrease the value of \(\hat{\gamma}_v\). By possibly repeating this argument, we can decrease \(\hat{\gamma}_v\) to \(1\), so that \(\hat{\gamma}_v + \hat{y}_v = 1\).

        By similar arguments we can show that the other ``without loss of generality''-assumptions hold as well.
    \end{proof}

    Now we can map an optimal solution for \eqref{eq: flow LP dual} to a solution for \eqref{eq: 2SAG LP eps} with the same objective value, as described above. This solution is optimal for \eqref{eq: 2SAG LP eps}, as otherwise we could find a better solution, map the better solution back to a solution for \eqref{eq: flow LP dual} with the same objective value, contradicting the optimality of the starting solution.
\end{proof}

\begin{lemma}
    \label{lem: 2SAG eps to 2SAG}
    An integral optimal solution for \eqref{eq: 2SAG LP eps} is also an integral optimal solution for \eqref{eq: 2SAG LP}.
\end{lemma}
\begin{proof}
    Let \(\hat{\gamma} = (\hat{y}, \hat{y}^S \text{ for } S \in S, \hat{\delta}^S \text{ for } S \in S, \hat{d}^S \text{ for } S \in S)\) be an integral optimal solution for \eqref{eq: 2SAG LP eps}.
    
    Suppose that \(1^\top \hat{y} > \nu(G_0)\), then because \(\hat{y}\) is integral, \(1^\top \hat{y} \geq \nu(G_0) + 1\). Now replace \(\hat{y}\) by a minimum vertex cover \(\widetilde{y}\), i.e., \(1^\top \widetilde{y} = \nu(G_0)\), and set \(\widetilde{\delta}^{S}\) and \(\widetilde{d}^S\) accordingly for all \(S \in \S\).
    We will have \(\widetilde{y} \leq 1\) and \(\hat{y}^S \leq 1\), and hence also \(\widetilde{\delta}^S \leq 1\) and \(\widetilde{d}^S \leq 1\). Therefore,
    \begin{equation*}
    \begin{split}
        \varepsilon \left( \sum_{v \in V_0} \alpha_v \widetilde{y}_v + \sum_{S \in \S} \sum_{v \in V_S} \alpha^S_v \hat{y}^S_v \right.&\left. + \sum_{S \in \S} \sum_{v \in V_0 \cap V_S} \beta^S_v \widetilde{\delta}^S_v + b^S_v \widetilde{d}^S_v \right)\\
        &\leq \varepsilon \left( \sum_{v \in V_0} \alpha_v + \sum_{S \in \S} \sum_{v \in V_S} \alpha^S_v + \sum_{S \in \S} \sum_{v \in V_0 \cap V_S} \beta^S_v + b^S_v  \right)\\
        &\leq \varepsilon \left( \sum_{v \in V_0} |\alpha_v| + \sum_{S \in \S} \sum_{v \in V_S} |\alpha^S_v| + \sum_{S \in \S} \sum_{v \in V_0 \cap V_S} \beta^S_v + b^S_v \right)\\
        &< \varepsilon \left( 1+ \sum_{v \in V_0} |\alpha_v| + \sum_{S \in \S} \sum_{v \in V_S} |\alpha^S_v| + \sum_{S \in \S} \sum_{v \in V_0 \cap V_S} \beta^S_v + b^S_v \right)\\
        &= 1,
    \end{split}
    \end{equation*}
    which means we obtain a strictly better solution, contradicting that \(\hat{\gamma}\) is optimal. So, \(1^\top \hat{y} = \nu(G_0)\). Similarly, \(1^\top \hat{y}^S = \nu(G_S)\) for all \(S \in \S\). Hence, \(\hat{\gamma}\) is feasible for \eqref{eq: 2SAG LP}.

    Suppose \(\hat{\gamma}\) is not optimal for \eqref{eq: 2SAG LP}. Let \(\widetilde{\gamma}\) be an optimal solution for \eqref{eq: 2SAG LP}. Then
    \begin{multline*}
        \sum_{v \in V_0} \alpha_v \widetilde{y}_v + \sum_{S \in \S} \sum_{v \in V_S} \alpha^S_v \widetilde{y}^S_v+ \sum_{S \in \S} \sum_{v \in V_0 \cap V_S} \beta^S_v \widetilde{\delta}^S_v + b^S_v \widetilde{d}^S_v\\
        < \sum_{v \in V_0} \alpha_v \hat{y}_v + \sum_{S \in \S} \sum_{v \in V_S} \alpha^S_v \hat{y}^S_v+ \sum_{S \in \S} \sum_{v \in V_0 \cap v_S} \beta^S_v \hat{\delta}^S_v + b^S_v \hat{d}^S_v,
    \end{multline*}
    \(1^\top \widetilde{y} = \nu(G_0) = 1^\top \hat{y}\), and \(1^\top \widetilde{y}^S = \nu(G_S) = 1^\top \hat{y}^S\) for all \(S \in \S\). It follows that
    \begin{multline*}
        1^\top \widetilde{y} + \sum_{S \in \S} 1^\top \widetilde{y}^S + \varepsilon \left( \sum_{v \in V_0} \alpha_v \widetilde{y}_v + \sum_{S \in \S} \sum_{v \in V_S} \alpha^S_v \widetilde{y}^S_v + \sum_{S \in \S} \sum_{v \in V_0 \cap V_S} \beta^S_v \widetilde{\delta}^S_v + b^S_v \widetilde{d}^S_v \right) \\
        < 1^\top \hat{y} + \sum_{S \in \S} 1^\top \hat{y}^S + \varepsilon \left( \sum_{v \in V_0} \alpha_v \hat{y}_v + \sum_{S \in \S} \sum_{v \in V_S} \alpha^S_v \hat{y}^S_v + \sum_{S \in \S} \sum_{v \in V_0 \cap V_S} \beta^S_v \hat{\delta}^S_v + b^S_v \hat{d}^S_v \right),
    \end{multline*}
    because \(\varepsilon > 0\). This contradicts the optimality of \(\hat{\gamma}\) for \eqref{eq: 2SAG LP eps}, hence \(\hat{\gamma}\) must be optimal for \eqref{eq: 2SAG LP}.
\end{proof}

Finally, \Cref{lem: integral polyhedron,lem: flow dual to 2SAG eps,lem: 2SAG eps to 2SAG} prove \Cref{thm: integral solution}.
\section{Implicit Distribution}
\label{sec:implicit}

We here prove that, when the distribution is not known, the problem becomes hard but it can still be well approximated using the SAA method. For the SAA analysis, the integrality result proved in the previous section plays a central role. In terms of techniques, the results in this section follow closely the ones in~\cite{Faenza2024Two-stage}, we still include all the details for the sake of completeness.

\subsection{Hardness}
\begin{theorem}
    When the second-stage distribution is specified implicitly by a sampling oracle, there exists no algorithm that solves \eqref{eq: 2SAG} in time polynomial in the input size and the number of calls to the oracle, unless \(\text{P} = \text{NP}\). This holds even if \(\lambda\) is nonzero for only one vertex \(v \in V_0\), and if all second-stage scenarios are obtained by only removing vertices.
\end{theorem}
\begin{proof}
    As in \cite{Faenza2024Two-stage}, we will prove this hardness result by showing that if such an algorithm were to exists, then it could be used to count the number of vertex covers in a graph in polynomial time. However, counting the number of vertex covers in a graph is \#P-hard \cite{Provan1983Complexity}.
    
    Let \(G = (V,E)\) be any undirected graph. We create an instance of \eqref{eq: 2SAG} as follows.
    \begin{itemize}
        \item \emph{First-stage instance:}
            The first-stage graph \(G_0= (V_0, E_0)\) is given by 
            \begin{equation*}
                V_0 = \cbra{v_1,\ldots,v_{d_v} : v \in V} \cup E \cup \cbra{\alpha, \beta_1, \beta_2},
            \end{equation*}
            and
            \begin{equation*}
                E_0 = \cbra{e v_1, \ldots, e v_{d_v} : v \in e \in E} \cup \cbra{\alpha e: e \in E} \cup \cbra{\alpha \beta_1, \alpha \beta_2}.
            \end{equation*}
            This is a bipartite graph with bipartitions \(\cbra{v_1,\ldots,v_{d_v} : v \in V} \cup \cbra{\alpha}\) and \(E \cup \cbra{\beta_1, \beta_2}\).
        \item \emph{Second-stage instance:}
            Sampling from the second-stage distribution \(\D\) consists of the following: Add the players \(\cbra{v_1, \ldots, v_{d_v}}\) with probability \(\frac12\), independently for all \(v \in V\). Add the players \(E \cup \alpha\). The second-stage graph is the subgraph of \(G_0\) induced by these vertices. Observe that this graph is bipartite, and in particular has the same bipartition as \(G_0\).
        \item \emph{Dissatisfaction costs:}
            Set \(\lambda = 0\) except for \(\lambda_\alpha = 1\).
    \end{itemize}

    In the first stage, both \(\beta_1\) and \(\beta_2\) only have an edge to \(\alpha\), which means that at least one of them will be exposed. We have \(\nu(G_0 \setminus \beta_i) = \nu(G_0)\) for \(i = 1,2\). It follows that in any core element, they have value zero. To cover the edges between \(\beta_1\), \(\beta_2\) and \(\alpha\), it follows that in any core element, \(\alpha\) must have value one. Also observe that in the second stage, \(\alpha\) will have core value at most one. The objective \eqref{eq: 2SAG} in this case becomes
    \begin{equation*}
        \E_{S \sim \D} \bra{ \min_{y^S \in \core(G_S)} 1 - y^S_\alpha }.
    \end{equation*}

    A second-stage vertex set will look like
    \begin{equation*}
        \cbra{v_1,\ldots,v_{d_v} : v \in S} \cup E \cup \alpha,
    \end{equation*}
    for some \(S \subseteq V\). Denote this set by \(\Pi(S)\). For a given \(S \subseteq V\), the probability that \(\Pi(S)\) is the second-stage vertex set is \(\frac{1}{2^{|V|}}\). With this information, we can write down the expectation explicitly:
    \begin{equation*}
        \sum_{S \subseteq V} \frac{1}{2^{|V|}} \bra{ \min_{y^S \in \core(G_0[\Pi(S)])} 1 - y^S_\alpha }.
    \end{equation*}

    Suppose \(S \subseteq V\) is a vertex cover of \(G\). Since \(S\) is a vertex cover, for each edge \(e = uv \in E\) at least one of \(u\) and \(v\) is in \(S\). Without loss of generality, let us assume that \(v \in S\). Then in \(G_0[\Pi(S)]\), the vertex \(e\) can be matched to any copy of \(v\); since we added \(d_v\) copies of \(v\), there are definitely enough copies to cover all edge-vertices. This matching is perfect on one side of the bipartition, which means that it is maximum. Since \(\alpha\) is not matched in this matching, we must have \(y^S_\alpha = 0\).

    Suppose \(S \subseteq V\) is not a vertex cover of \(G\). Since \(S\) is not a vertex cover, there exists an edge \(e = uv \in E\) such that neither \(u\) nor \(v\) is in \(S\). Consequently, in \(G_0[\Pi(S)]\), the vertex \(e\) only has an edge to \(\alpha\). To cover this edge with the core element, we must have \(y^S_e + y^S_\alpha = 1\). Since \(e\) is not incident with any other edges, it is feasible to set \(y^S_e = 0\) and \(y^S_\alpha = 1\). This is also the solution that minimizes \(1 - y^S_\alpha\).

    From these arguments it follows that
    \begin{subequations}
    \begin{align}
        \sum_{S \subseteq V} \frac{1}{2^{|V|}} \bra{ \min_{y^S \in \core(G_0[S])} 1 - y^S_\alpha }
         &= \frac{1}{2^{|V|}} ( \text{\#vertex covers} (1-0) + \text{\#not vertex covers} (1-1) ), \\
         \label{eq: nr vertex covers}
         &= \frac{1}{2^{|V|}} \text{\#vertex covers}.
    \end{align}
    \end{subequations}
    So, if we could solve \eqref{eq: 2SAG} in polynomial-time, i.e., determine its optimal objective value (\(=\) \eqref{eq: nr vertex covers}), then we could also determine the number of vertex covers in any graph in polynomial-time. 
\end{proof}

\subsection{SAA Algorithm}
    The \emph{sample average approximation (SAA)} method is a well-known method in stochastic programming. It has been exploited often to approximate two-stage stochastic combinatorial problems \cite{Charikar2005Sampling,Swamy2008Sample,Swamy2012Sampling-based,Faenza2024Two-stage,Ravi2006Hedging}.
    Let \(S^1,\ldots,S^N\) be \(N\) i.i.d.\ samples drawn from the distribution \(\D\). We replace the objective function in \eqref{eq: 2SAG} by the average taken over our samples \(S^1,\ldots,S^N\):
    \begin{equation}
        \label{eq: SAA}
        \tag{SAA}
        \min_{y \in \core(G_0)} \frac1N \sum_{i=1}^N \bra{ \min_{y^i \in \core(G_i)} \sum_{v \in V_0 \cap V_i} \lambda_v \bra{y_v - y^i_v}^+ },
    \end{equation}
    where we use \(G_i = (V_i, E_i)\) to denote \(G_{S^i} = (V_{S^i}, E_{S^i})\) for \(i = 1, \ldots, N\).
    Observe that \eqref{eq: SAA} is an instance of \eqref{eq: 2SAG explicit}, where \(\S = \cbra{S^1,\ldots,S^N}\) and \(p_S = \frac1N\) for all \(S \in \S\), which means we can solve this problem by solving \eqref{eq: 2SAG explicit}, and in particular we can obtain an integral solution.

    For any instance \(\I\) of \eqref{eq: 2SAG}, we denote by \(y^\I\) the optimal solution for instance \(\I\), and by \(\val_\I(y)\) the objective value of \(y\) in instance \(\I\).

    \begin{theorem}
        \label{thm: SAA algorithm}
        Given an instance \(\I\) of \eqref{eq: 2SAG} where a sampling oracle specifies the second-stage distribution implicitly, and two parameters \(\epsilon > 0\), \(\alpha \in (0,1)\),
        one can compute a first-stage core element \(y\) such that
        \begin{equation*}
            \P(\val_\I(y) \leq \val_\I(y^\I) + \varepsilon) \geq 1 - \alpha,
        \end{equation*}
        in time polynomial in the size of \(\I\), \(\lambda\), \(\ln(1/\alpha)\) and \(1/\varepsilon\).
    \end{theorem}

    Let \(\hat{y}\) be an extreme point optimal solution for \eqref{eq: SAA}, which is integral by \Cref{thm: integral solution}. As in \cite{Faenza2024Two-stage}, the main ingredient to prove \Cref{thm: SAA algorithm} is the next lemma.
    
    \begin{lemma}
        \label{lem:SAA}
        For any \(\alpha \in (0,1)\), the following holds with probability at least \(1-\alpha\)
        \begin{equation*}
            \val_\I(\hat{y}) \leq \val_\I(y^\I) + \sqrt{2} \sum_{v \in V_0} \lambda_v \sqrt{\frac{\ln(2^{|V_0|}/\alpha)}{N}}.
        \end{equation*}
    \end{lemma}
    The proof of this lemma follows the proof of the corresponding lemma in \cite{Faenza2024Two-stage} closely. The key difference is that where \cite{Faenza2024Two-stage} uses that the number of stable matchings is bounded, we use the integrality result of previous section to bound the amount of core solutions we have to consider. We denote by \(\val_\SAA(y)\) the objective value of \(y\) in an instance of \eqref{eq: SAA}.
    \begin{proof}[Proof of \Cref{lem:SAA}]
        For \(i \in \cbra{1,\ldots,N}\) and \(y \in \core(G_0)\), let
        \begin{equation*}
            G_i(y) = \min_{y^i \in \core(G_i)} \sum_{v \in V_0 \cap V_i} \lambda_v \bra{y_v - y^i_v}^+.
        \end{equation*}
        Let
        \begin{equation*}
            \varepsilon = \sqrt{2} \sum_{v \in V_0} \lambda_v \sqrt{\frac{\ln(2^{|V_0|}/\alpha)}{N}},
        \end{equation*}
        and let \(\F^\varepsilon = \cbra{y \in \core(G_0) : \val_\I(y) \leq \val_\I(y^\I) + \varepsilon}\). We obtain
        \begin{equation*}
        \begin{split}
            \P(\hat{y} \notin \F^\varepsilon)
                &\leq \sum_{y \notin \F^\varepsilon, y \in \{0,1\}^{|V_0|}} \P(y \text{ is an optimal solution for \eqref{eq: SAA}}), \\
                &\leq \sum_{y \notin \F^\varepsilon, y \in \{0,1\}^{|V_0|}} \P\Par{\val_\SAA(y) \leq \val_\SAA(y^\I)}, \\
                &= \sum_{y \notin \F^\varepsilon, y \in \{0,1\}^{|V_0|}} \P\Par{\frac1N \sum_{i=1}^N G_i(y) \leq \frac1N \sum_{i=1}^N G_i(y^\I)}, \\
                &= \sum_{y \notin \F^\varepsilon, y \in \{0,1\}^{|V_0|}} \P\Par{\frac1N \sum_{i=1}^N \Par{G_i(y) - G_i(y^\I)} \leq 0}. \\
        \end{split}
        \end{equation*}
        Observe that we can restrict ourselves to sum over integral \(y\) not in \(\F^\varepsilon\), as \(\hat{y}\) is integral. 
        Fix \(y \notin \F^\varepsilon\), \(y \in \{0,1\}^{|V_0|}\) and let \(X_i = G_i(y) - G_i(y^\I)\). We will bound the probability \(\P(\frac1N \sum_{i=1}^N X_i \leq 0)\) using the classical Hoeffding's inequality, stated below.
        \begin{lemma}
            [\cite{Hoeffding1994Probability,Faenza2024Two-stage}]
            Let \(X_1,\ldots,X_n\) be independent random variables such that \(a_i \leq X_i \leq b_i\) almost surely. Consider the sum of these random variables \(S_n = X_1 + \cdots + X_n\). The Hoeffding's inequality states that for all \(t > 0\),
            \begin{equation*}
                \P(\E(S_n) - S_n \geq t) \leq \exp\Par{\frac{-2t^2}{\sum_{i=1}^n (b_i - a_i)^2}}.
            \end{equation*}
        \end{lemma}
        It follows from \(0 \leq y \leq 1\) that \(0 \leq [y_v - y^i]^+ \leq 1\), and consequently that
        \begin{equation*}
            0 \leq G_i(y) \leq \sum_{v \in V_0 \cap V_i} \lambda_v \leq \sum_{v \in V_0} \lambda_v.
        \end{equation*}
        Let \(a_i = - \sum_{v \in V_0} \lambda_v\) and \(b_i = \sum_{v \in V_0} \lambda_v\). Then \(a_i \leq X_i \leq b_i\). Now since \(y \notin \F^\varepsilon\), it holds that \(\E(X_i) \geq \varepsilon\). Let \(S_N = \sum_{i=1}^N X_i\). We obtain
        \begin{equation*}
        \begin{split}
            \P\Par{\frac1N \sum_{i=1}^N X_i \leq 0}
                &= \P\Par{\frac1N S_N \leq 0} 
                = \P\Par{S_N \leq 0}, \\
                &= \P\Par{\E(S_N) - S_N \geq \E(S_N)}, \\
                &\leq \P\Par{\E(S_N) - S_N \geq \varepsilon N}, \\
                &\leq \exp\Par{\frac{-2 \varepsilon^2 N^2}{\sum_{i=1}^N (b_i - a_i)^2}}, \\
                &= \exp\Par{\frac{-2^2 \Par{\sum_{v \in V_0} \lambda_v}^2 \ln(2^{|V_0|}/\alpha) N}{N (2 \sum_{v \in V_0} \lambda_v)^2}}, \\
                &= \exp\Par{-\ln(2^{|V_0|}/\alpha)} = \frac{\alpha}{2^{|V_0|}}.
        \end{split}
        \end{equation*}
        So finally,
        \begin{equation*}
            \P(\hat{y} \notin \F^\varepsilon)
                \leq \sum_{y \notin \F^\varepsilon, y \in \{0,1\}^{|V_0|}} \P\Par{\frac1N \sum_{i=1}^N X_i \leq 0}
                \leq \sum_{y \notin \F^\varepsilon, y \in \{0,1\}^{|V_0|}} \frac{\alpha}{2^{|V_0|}}
                \leq \alpha,
        \end{equation*}
        where the last inequality follows from the fact that since \(y \in \cbra{0,1}^{|V_0|}\), there are at most \(2^{|V_0|}\) terms in the sum.
    \end{proof}
    
    Finally, for any \(\varepsilon > 0\) and \(\alpha \in (0,1)\), setting \(N = 2 \Par{\sum_{v \in V_0} \lambda_v}^2 \ln(2^{|V_0|}/\alpha) / \varepsilon^2\) proves \Cref{thm: SAA algorithm}.
\section{Multistage Setting}
\label{sec:multistage}

In this section we consider a multistage setting where there are \(k\) stages, with predetermined graphs (no distribution). We denote by \(G_i = (V_i, E_i)\) the graph of the \(i\)'th stage for \(i = 1, \ldots, k\).
This setting resembles the setting in \cite{Fluschnik2022Mutlistage}, who discuss multistage vertex cover. As before, and as in \cite{Fluschnik2022Mutlistage}, we will consider the problem of minimizing the absolute difference. Again, one can still choose later to minimize the positive difference by choosing which variables to take into the objective. We can formulate this problem as follows.
\begin{equation}
    \label{eq:multistage}
    \min_{y^i \in \core(G_i), i = 1, \ldots, k} \sum_{i = 1}^{k-1} \sum_{v \in V_i \cap V_{i+1}} \lambda^i_v \Abs{y^i_v - y^{i+1}_v}.
\end{equation}
Like before, we can formulate this as the following LP.
\begin{equation}
\label{eq:multistage_LP}
\begin{split}
    \min \quad 
        & \sum_{i = 1}^{k-1} \sum_{v \in V_i \cap V_{i+1}} \lambda^i_v \Par{\delta^i_v + d^i_v} \\
    \text{s.t.} \quad
        & y^i_u + y^i_v \geq 1 \quad \forall uv \in E_i, \forall i = 1, \ldots, k \\
        & 1^\top y^i = \nu(G_i) \quad \forall i = 1, \ldots, k \\
        & y^i \in \R^{V_i}_{\geq 0} \quad \forall i = 1, \ldots, k \\
        & y^i_v -  y^{i+1}_v \leq \delta^i_v \quad \forall v \in V_i \cap V_{i+1}, \forall i = 1, \ldots, k-1 \\
        & y^{i+1}_v - y^i_v \leq d^i_v \quad \forall v \in V_i \cap V_{i+1}, \forall i = 1, \ldots, k-1 \\
        & \delta^i \in \R^{V_i \cap V_{i+1}}_{\geq0} \quad \forall i = 1, \ldots, k-1 \\
        & d^i \in \R^{V_i \cap V_{i+1}}_{\geq0} \quad \forall i = 1, \ldots, k-1
\end{split}
\end{equation}
Let \(V = \cup_{i=1}^k V_i\).
This LP has \(O(|V| k)\) variables and \(O(|V|^2 k)\) constraints. So it has size polynomial in the input size, which means we can solve \eqref{eq:multistage} in polynomial-time, by solving \eqref{eq:multistage_LP}. 
Like in \Cref{sec:explicit}, we can show that the feasible region of \eqref{eq:multistage_LP} is an integral polyhedron.

\begin{theorem}
    \label{thm:multistage_int_solution}
    The feasible region of \eqref{eq:multistage_LP} is an integral polyhedron.
\end{theorem}

Before we go into the proof, we discuss a consequence of this theorem. In particular, we can use this theorem to
prove a result about the \emph{multistage} vertex cover problem, which we now define formally. A vertex cover \(C\) in a graph \(G = (V, E)\) is a subset of vertices \(C \subseteq V\) such for each edge \(uv \in E\), we have \(\{u, v\} \cap C \neq \emptyset\). In the multistage vertex cover problem there are \(k\) stages with predetermined graphs \(G_1, \ldots, G_k\). The goal is to find a vertex cover for each stage \(C_1, \ldots, C_k\) such that the total absolute difference between stages is minimized, i.e. 
\begin{equation*}
    \sum_{i=1}^{k-1} |C_i \setminus C_{i+1}| + |C_{i+1} \setminus C_{i}|.
\end{equation*}

The authors in \cite{Fluschnik2022Mutlistage} have shown that the multistage vertex cover problem is NP-hard, already when \(k = 2\), the first-stage graph is a path and the second-stage graph is a tree. Both graphs are indeed very simple bipartite graphs. Our result shows that the difficulty lies in the fact that the bipartitions for the two stages are different. In fact, with the additional requirement that the bipartitions are the same, we prove that the problem becomes polynomial time solvable for any number of stages.

\begin{theorem}
    The multistage vertex cover problem is solvable in polynomial time when all \(G_i\), \(i = 1, \ldots, k\), are bipartite with the same bipartition.
\end{theorem}
\begin{proof}
    We can solve the multistage vertex cover problem on \(G_i\), \(i = 1, \ldots, k\) by modeling it as a multistage assignment game with the LP in~\eqref{eq:multistage_LP}. There, we set  \(\lambda^i_v = 1\) for all \(i = 1, \ldots, k-1\) and all \(v \in V_i \cap V_{i+1}\). By \Cref{thm:multistage_int_solution} we can obtain an integral optimal solution, which means that the vectors \(y^i\) in this solution indicate vertex covers \(C_i\).
\end{proof}

Now we go into the proof of \Cref{thm:multistage_int_solution}, following the same line of argument as in the proof of \Cref{thm: integral solution}.

We will consider \eqref{eq:multistage_LP} with an arbitrary objective, and show that for each objective we can find an integral optimal solution. Since for each extreme point of a polyhedron there is an objective such that the extreme point is the unique optimal solution, this proves that the polyhedron is integral. We consider the following objective, where \(\alpha\) can take any value, \(\beta \geq 0\) and \(b \geq 0\). Note that if \(\beta^i_v < 0\) or \(b^i_v < 0\) for any \(v \in V_i \cap V_{i+1}\), \(i = 1, \ldots, k-1\), then the LP becomes unbounded, and so there are no extreme points in those directions.
\begin{equation*}
    \min \quad 
        \sum_{i=1}^k \sum_{v \in V_i} \alpha^i_v y_v+ \sum_{i=1}^{k-1} \sum_{v \in V_i \cap V_{i+1}} \beta^i_v \delta^i_v + b^i_v d^i_v
\end{equation*}

We will formulate a maximum flow problem on an auxiliary graph, and compute its dual LP. It is known that the constraint matrix of a maximum flow LP is always totally unimodular (TU). The dual LP has the transpose as constraint matrix, which is then also TU. Since in addition the right hand side of the dual constraints are integral, this means the feasible region of the dual is an integral polyhedron. (See e.g.~\cite{Schrijver2003Combinatorial} for properties of TU matrices.) Finally, we show that we can map an optimal dual solution to an optimal solution for \eqref{eq:multistage_LP}.

Let \(V_s, V_t \subseteq \cup_{i=1}^k V_i\) be the bipartition of all \(G_i\) combined. Let
\begin{equation*}
    \varepsilon = \frac{1}{1 + \sum_{v \in V_i} |\alpha^i_v| + \sum_{i=1}^{k-1} \sum_{v \in V_i \cap V_{i+1}} \beta^i_v + b^i_v}.
\end{equation*}
Note that \(\varepsilon > 0\).
We create the auxiliary graph \(G' = (V', A)\) as follows. Let \(V' = \cup_{i=1}^k V_i \cup \cbra{s, t}\). To make clear about which version of a vertex \(v \in V_i \cap V_{i+1}\) we are talking, we will (sometimes) denote them by \(v^i\) and \(v^{i+1}\).
The arc set \(A\) contains the following arcs:
\begin{itemize}
    \item \(sv\) for all \(v \in V_i \cap V_s\) and \(i = 1, \ldots, k\), with flow-capacity \(1 + \varepsilon \alpha^i_v\);
    \item \(vt\) for all \(v \in V_i \cap V_t\) and \(i = 1, \ldots, k\), with flow-capacity \(1 + \varepsilon \alpha^i_v\);
    \item \(v^{i+1} v^i\) for all \(v \in (V_i \cap V_{i+1}) \cap V_s\), \(i = 1, \ldots, k-1\), with flow-capacity \(\varepsilon \beta^i_v\);
    \item \(v^i v^{i+1}\) for all \(v \in (V_i \cap V_{i+1}) \cap V_s\), \(i = 1, \ldots, k-1\), with flow-capacity \(\varepsilon b^i_v\);
    \item \(v^i v^{i+1}\) for all \(v \in (V_i \cap V_{i+1}) \cap V_t\), \(i = 1, \ldots, k-1\), with flow-capacity \(\varepsilon \beta^i_v\);
    \item \(v^{i+1} v^i\) for all \(v \in (V_i \cap V_{i+1}) \cap V_t\), \(i = 1, \ldots, k-1\), with flow-capacity \(\varepsilon b^i_v\);
    \item \(uv\) for all \(uv \in \bigcup_{i=1}^k E_i\), such that \(u \in V_s\), \(v \in V_t\), without upperbound on the flow-capacity.
\end{itemize}

We can formulate the maximum flow problem in \(G'\) as an LP as follows.
\begin{equation}
\label{eq:multistage_flow_LP}
\begin{split}
    \max \quad 
        & \sum_{v \in V_s} f_{sv} \\
    \text{s.t.} \quad
        & f_{sv} + \1[v \in V_{i-1}] (f_{v^{i-1} v} - f_{v v^{i-1}}) + \1[v \in V_{i+1}] (f_{v^{i+1} v} - f_{v v^{i+1}}) - \sum_{u: uv \in E_i} f_{vu} = 0 \\&\hspace{9cm}\quad \forall v \in V_i \cap V_s, \forall i = 1, \ldots, k \\
        & \sum_{u: uv \in E_i} f_{uv} + \1[v \in V_{i-1}] (f_{v^{i-1} v} - f_{v v^{i-1}}) + \1[v \in V_{i+1}] (f_{v^{i+1} v} - f_{v v^{i+1}}) - f_{vt} = 0 \\&\hspace{9cm}\quad \forall v \in V_i \cap V_t, \forall i = 1, \ldots, k \\
        & f_{sv} \leq 1 + \varepsilon \alpha^i_v \quad \forall v \in V_i \cap V_s, \forall i = 1, \ldots, k \\
        & f_{vt} \leq 1 + \varepsilon \alpha^i_v \quad \forall v \in V_i \cap V_t, \forall i = 1, \ldots, k \\
        & f_{v^{i+1} v^i} \leq \varepsilon \beta^i_v \quad \forall v \in (V_i \cap V_{i+1}) \cap V_s, \forall i = 1, \ldots, k-1 \\
        & f_{v^i v^{i+1}} \leq \varepsilon b^i_v \quad \forall v \in (V_i \cap V_{i+1}) \cap V_s, \forall i = 1, \ldots, k-1 \\
        & f_{v^i v^{i+1}} \leq \varepsilon \beta^i_v \quad \forall v \in (V_i \cap V_{i+1}) \cap V_t, \forall i = 1, \ldots, k-1 \\
        & f_{v^{i+1} v^i} \leq \varepsilon b^i_v \quad \forall v \in (V_i \cap V_{i+1}) \cap V_t, \forall i = 1, \ldots, k-1 \\
        & f \in \R^A_{\geq 0}
\end{split}
\end{equation}
For this to have a feasible flow, the flow-capacities need to be nonnegative: \(1 + \varepsilon \alpha^i_v \geq 0\) for all \(v \in V_i\), \(i = 1, \ldots, k\), \(\varepsilon \beta^i_v \geq 0\) for all \(v \in V_i \cap V_{i+1}\), \(i = 1, \ldots, k-1\), and \(\varepsilon b^i_v \geq 0\) for all \(v \in V_i \cap V_{i+1}\), \(i = 1, \ldots, k-1\). The later two are satisfied as \(\varepsilon > 0\) and we set \(\beta, b \geq 0\). For any \(v \in V_i\), \(i = 1, \ldots, k\), we have
\begin{equation*}
    \alpha^i_v 
    \geq - |\alpha^i_v|
    \geq - \left( 1 + \sum_{v \in V_i} |\alpha^i_v| + \sum_{i=1}^{k-1} \sum_{v \in V_i \cap V_{i+1}} \beta^i_v + b^i_v \right)
    = -\frac{1}{\varepsilon},
\end{equation*}
and so, since \(\varepsilon > 0\), we have \(1 + \varepsilon \alpha^i_v \geq 1 + \varepsilon \frac{-1}{\varepsilon} = 0\). 

The dual of this flow LP is as follows.
\begin{equation}
\label{eq:multistage_flow_LP_dual}
\begin{split}
    \min \quad 
        & \sum_{i=1}^k \sum_{v \in V_i} y^i_v + \varepsilon \left( \sum_{i=1}^k \sum_{v \in V_i} \alpha^i_v y^i_v + \sum_{i=1}^{k-1} \sum_{v \in V_i \cap V_{i+1}} \beta^i_v \delta^i_v + b^i_v d^i_v \right) \\
    \text{s.t.} \quad
        & \gamma^i_v + y^i_v \geq 1 \quad \forall v \in V_i \cap V_s, \forall i = 1, \ldots, k \\
        & \gamma^i_v - \gamma^{i+1}_v + \delta^i_v \geq 0 \quad \forall v \in (V_i \cap V_{i+1}) \cap V_s, \forall i = 1, \ldots, k-1 \\
        & \gamma^{i+1}_v - \gamma^i_v + d^i_v \geq 0 \quad \forall v \in (V_i \cap V_{i+1}) \cap V_s, \forall i = 1, \ldots, k-1 \\
        & \gamma^i_v - \gamma^i_u \geq 0  \quad \forall uv \in E_i \text{ such that } u \in V_s, v \in V_t, \forall i = 1, \ldots, k \\
        & - \gamma^i_v + y^i_v \geq 0 \quad \forall v \in V_i \cap V_t, \forall i = 1, \ldots, k \\
        & \gamma^{i+1}_v - \gamma^i_v + \delta^i_v \geq 0 \quad \forall v \in (V_i \cap V_{i+1}) \cap V_t, \forall i = 1, \ldots, k-1 \\
        & \gamma^i_v - \gamma^{i+1}_v + d^i_v \geq 0 \quad \forall v \in (V_i \cap V_{i+1}) \cap V_t, \forall i = 1, \ldots, k-1 \\
        & y^i \in \R^{V_i}_{\geq 0} \quad \forall i = 1, \ldots, k\\
        & \delta^i \in \R^{V_i \cap V_{i+1}}_{\geq 0} \quad \forall i = 1, \ldots, k-1 \\
        & d^i \in \R^{V_i \cap V_{i+1}}_{\geq 0} \quad \forall i = 1, \ldots, k-1 \\
        & \gamma^i \in \R^{V_i} \quad \forall i = 1, \ldots, k
\end{split}
\end{equation}

\begin{lemma}
    \label{lem:multistage_int_polyhedron}
    The feasible region of \eqref{eq:multistage_flow_LP_dual} is an integral polyhedron.
\end{lemma}
\begin{proof}
    The constraint matrix of a maximum flow LP is always TU; in particular the constraint matrix of \eqref{eq:multistage_flow_LP} is TU.

    The transpose of a TU matrix is TU. The constraint matrix of a dual LP is the transpose of the constraint matrix of the primal. Together they imply that the constraint matrix of \eqref{eq:multistage_flow_LP_dual} is TU.

    Finally, a polyhedron with a TU constraint matrix and integral right hand side is an integral polyhedron; in particular, the feasible region of \eqref{eq:multistage_flow_LP_dual} is an integral polyhedron.
\end{proof}

We map an integral optimal solution for \eqref{eq:multistage_flow_LP_dual} to an integral optimal solution for \eqref{eq:multistage_LP} in two steps. First we map it to the following LP.
\begin{equation}
\label{eq:multistage_LP_eps}
\begin{split}
\min \quad 
    & \sum_{i=1}^k 1^\top y^i + \varepsilon \left( \sum_{i=1}^k \sum_{v \in V_i} \alpha^i_v y^i_v + \sum_{i=1}^{k-1} \sum_{v \in V_i \cap V_{i+1}} \beta^i_v \delta^i_v + b^i_v d^i_v \right) \\
\text{s.t.} \quad
    & y^i_u + y^i_v \geq 1 \quad \forall uv \in E_i, \forall i = 1, \ldots, k \\
    & y^i \in \R^{V_i}_{\geq 0} \quad \forall i = 1, \ldots, k \\
    & y^i_v -  y^{i+1}_v \leq \delta^i_v \quad \forall v \in V_i \cap V_{i+1}, \forall i = 1, \ldots, k-1 \\
    & y^{i+1}_v - y^i_v \leq d^i_v \quad \forall v \in V_i \cap V_{i+1}, \forall i = 1, \ldots, k-1 \\
    & \delta^i \in \R^{V_i \cap V_{i+1}}_{\geq0} \quad \forall i = 1, \ldots, k-1 \\
    & d^i \in \R^{V_i \cap V_{i+1}}_{\geq0} \quad \forall i = 1, \ldots, k-1
\end{split}
\end{equation}

\begin{lemma}
    \label{lem:flow_dual_to_multistage_eps}
    We can obtain an optimal solution for \eqref{eq:multistage_LP_eps} from an optimal solution for \eqref{eq:multistage_flow_LP_dual}.
\end{lemma}
\begin{proof}
    \begin{claim}
        Any feasible solution for \eqref{eq:multistage_LP_eps} can be mapped to a feasible solution for \eqref{eq:multistage_flow_LP_dual} with the same objective value.
    \end{claim}
    \begin{proof}
        Let \(\hat{\mu} = (\hat{y}^i \text{ for } i = 1, \ldots, k, \hat{\delta}^i \text{ for } i = 1, \ldots, k-1, \hat{d}^i \text{ for } i = 1, \ldots, k-1)\) be a feasible solution for \eqref{eq:multistage_LP_eps}. We extend it to a solution \(\hat{\mu}_{\text{ext}}\) for \eqref{eq:multistage_flow_LP_dual} by setting \(\hat{\gamma}^i_v = 1 - \hat{y}^i_v\) for all \(v \in V_i \cap V_s\), \(i = 1, \ldots, k\), and \(\hat{\gamma}^i_v = \hat{y}^i_v\) for all \(v \in V_i \cap V_t\), \(i = 1, \ldots, k\). 
        It is clear that \(\hat{\mu}_{\text{ext}}\) has the same objective value as \(\hat{\mu}\), as the objective functions of the two linear programs are the same and do not involve the \(\gamma\) variables. We next show that \(\hat{\mu}_{\text{ext}}\) is feasible.

        The constraints \(\gamma^i_v + y^i_v \geq 1\) and \(- \gamma^i_v + y^i_v \geq 0\) of \eqref{eq:multistage_flow_LP_dual} are satisfied by \(\hat{\mu}_{\text{ext}}\) by definition. Nonnegativity of \(y^i\) for \(i = 1, \ldots, k\), and of \(\delta^i\) and \(d^i\) for \(i = 1, \ldots, k-1\) are satisfied by \(\hat{\mu}_{\text{ext}}\), because \(\hat{\mu}\) is feasible for \eqref{eq:multistage_LP_eps}.

        Let \(i = 1, \ldots, k-1\) and \(v \in (V_i \cap V_{i+1}) \cap V_s\). We have
        \begin{equation*}
            \hat{\gamma}^i_v - \hat{\gamma}^{i+1}_v + \hat{\delta}^i_v
            = 1 - \hat{y}^i_v - (1 - \hat{y}^{i+1}_v) + \hat{\delta}^i_v
            = \hat{y}^{i+1}_v - \hat{y}^i_v + \hat{\delta}^i_v
            \geq 0,
        \end{equation*}
        and
        \begin{equation*}
            \hat{\gamma}^{i+1}_v - \hat{\gamma}^i_v + \hat{d}^i_v
            = 1 - \hat{y}^{i+1}_v - (1 - \hat{y}^i_v) + \hat{d}^i_v
            = \hat{y}^i_v - \hat{y}^{i+1}_v + \hat{d}^i_v
            \geq 0,
        \end{equation*}
        where both times the last inequality follows from \(\hat{\mu}\)'s feasibility for \eqref{eq:multistage_LP_eps}. Similarly, we have for \(i = 1, \ldots, k-1\) and \(v \in (V_i \cap V_{i+1}) \cap V_t\)
        \begin{equation*}
            \hat{\gamma}^{i+1}_v - \hat{\gamma}^i_v + \hat{\delta}^i_v
            = \hat{y}^{i+1}_v - \hat{y}^i_v + \hat{\delta}^i_v
            \geq 0,
        \end{equation*}
        and
        \begin{equation*}
            \hat{\gamma}^i_v - \hat{\gamma}^{i+1}_v + \hat{d}^i_v
            = \hat{y}^i_v - \hat{y}^{i+1}_v + \hat{d}^i_v
            \geq 0.
        \end{equation*}
        Finally, let \(i = 1, \ldots, k\) and \(uv \in E_i\) such that \(u \in V_s\) and \(v \in V_t\). We have
        \begin{equation*}
            \hat{\gamma}^i_v - \hat{\gamma}^i_u
            = \hat{y}^i_v - (1 - \hat{y}^i_u)
            = \hat{y}^i_v + \hat{y}^i_u - 1
            \geq 0.
        \end{equation*}
    \end{proof}

    \begin{claim}
        Any feasible solution to \eqref{eq:multistage_flow_LP_dual} can be mapped to a feasible solution to \eqref{eq:multistage_LP_eps} with the same objective value.
    \end{claim}
    \begin{proof}
        Let \(\hat{\mu} = (\hat{y}^i \text{ for } i = 1, \ldots, k, \hat{\delta}^i \text{ for } i = 1, \ldots, k-1, \hat{d}^i \text{ for } i = 1, \ldots, k-1, \hat{\gamma}^i \text{ for } i = 1, \ldots, k)\) be a feasible solution for \eqref{eq:multistage_flow_LP_dual}. We restrict it to a solution \(\hat{\mu}_{\text{res}}\) for \eqref{eq:multistage_LP_eps} by disregarding the \(\gamma\) variables. It is clear that \(\hat{\mu}_{\text{res}}\) has the same objective value as \(\hat{\mu}\), as the objective functions of the two linear programs are the same and do not involve the \(\gamma\) variables. We next show that \(\hat{\mu}_{\text{res}}\) is feasible.

        First observe that \(\hat{\mu}_{\text{res}} \geq 0\).

        Let \(i = 1, \ldots, k\) and \(uv \in E_i\) such that \(u \in V_s\) and \(v \in V_t\). We have
        \begin{equation*}
            \hat{y}^i_u + \hat{y}^i_v
            = (\hat{\gamma}^i_u + \hat{y}^i_u) + (\hat{\gamma}^i_v - \hat{\gamma}^i_u) + (- \hat{\gamma}^i_v + \hat{y}^i_v)
            \geq 1 + 0 + 0 = 1.
        \end{equation*}

        Let \(i = 1, \ldots, k\), \(v \in (V_i \cap V_{i+1}) \cap V_s\). We will show that without loss of generality we can assume that \(\hat{\gamma}^j_u + \hat{y}^j_u = 1\) for \(u \in V_j \cap V_S\), \(j = 1, \ldots, k\), so in particular for \(u = v\) and \(j = i\), \(j = i + 1\). Consequently, we have
        \begin{equation*}
        \begin{split}
            \hat{\delta}^i_v + \hat{y}^{i+1}_v - \hat{y}^i_v
            &= \hat{\delta}^i_v + \hat{y}^{i+1}_v - \hat{y}^i_v + \hat{\gamma}^{i+1}_v - \hat{\gamma}^{i+1}_v + \hat{\gamma}^i_v - \hat{\gamma}^i_v \\
            &= (\hat{\delta}^i_v + \hat{\gamma}^i_v - \hat{\gamma}^{i+1}_v) + (\hat{\gamma}^{i+1}_v + \hat{y}^{i+1}_v) - (\hat{\gamma}^i_v + \hat{y}^i_v) \\
            & \geq 0 + 1 - 1 = 0,
        \end{split}
        \end{equation*}
        and
        \begin{equation*}
        \begin{split}
            \hat{d}^i_v + \hat{y}^i_v - \hat{y}^{i+1}_v
            &= \hat{d}^i_v + \hat{y}^i_v - \hat{y}^{i+1}_v + \hat{\gamma}^i_v - \hat{\gamma}^i_v + \hat{\gamma}^{i+1}_v - \hat{\gamma}^{i+1}_v \\
            &= (\hat{d}^i_v + \hat{\gamma}^{i+1}_v - \hat{\gamma}^i_v) + (\hat{\gamma}^i_v + \hat{y}^i_v) - (\hat{\gamma}^{i+1}_v + \hat{y}^{i+1}_v) \\
            & \geq 0 + 1 - 1 = 0.
        \end{split}
        \end{equation*}
        Now, let \(i = 1, \ldots, k\), \(v \in (V_i \cap V_{i+1}) \cap V_t\). We will also show that without loss of generality we can assume that \(- \hat{\gamma}^j_u + \hat{y}^j_u= 0\) for \(u \in V_j \cap V_t\), \(j = 1, \ldots, k\). Consequently, we have
        \begin{equation*}
        \begin{split}
            \hat{\delta}^i_v + \hat{y}^{i+1}_v - \hat{y}^i_v
            &= \hat{\delta}^i_v + \hat{y}^{i+1}_v - \hat{y}^i_v + \hat{\gamma}^{i+1}_v - \hat{\gamma}^{i+1}_v + \hat{\gamma}^i_v - \hat{\gamma}^i_v \\
            &= (\hat{\delta}^i_v + \hat{\gamma}^{i+1}_v - \hat{\gamma}^i_v) + (- \hat{\gamma}^{i+1}_v + \hat{y}^{i+1}_v) - (- \hat{\gamma}^i_v + \hat{y}^i_v) \\
            & \geq 0 + 0 - 0 = 0,
        \end{split}
        \end{equation*}
        and
        \begin{equation*}
        \begin{split}
            \hat{d}^i_v + \hat{y}^i_v - \hat{y}^{i+1}_v
            &= \hat{d}^i_v + \hat{y}^i_v - \hat{y}^{i+1}_v + \hat{\gamma}^i_v - \hat{\gamma}^i_v + \hat{\gamma}^{i+1}_v - \hat{\gamma}^{i+1}_v \\
            &= (\hat{d}^i_v + \hat{\gamma}^i_v - \hat{\gamma}^{i+1}_v) + (- \hat{\gamma}^i_v + \hat{y}^i_v) - (- \hat{\gamma}^{i+1}_v + \hat{y}^{i+1}_v) \\
            & \geq 0 + 0 - 0 = 0,
        \end{split}
        \end{equation*}
        This finishes the feasibility proof of \(\hat{\mu}_{\text{res}}\).

        To show that we can indeed assume without loss of generality that \(\hat{\gamma}^i_v + \hat{y}^i_v = 1\) for \(i = 1, \ldots, k\), \(v \in V_i \cap V_s\), suppose it does not hold, so: \(\hat{\gamma}^i_v + \hat{y}^i_v > 1\) for some \(i \in \{1, \ldots, k\}\) and \(v \in V_i \cap V_s\). In particular, let \(\hat{\gamma}^i_v + \hat{y}^i_v > 1\) for all \(i\) in a range, i.e., for all \(i \in \{j, \ldots, j+l\}\) for some \(j \in \{1, \ldots, k\}\) and integral \(l \geq 0\), such that either \(v \notin V_{j-1}\) or \(\hat{\gamma}^{j-1}_v + \hat{y}^{j-1}_v = 1\), likewise for \(j + l + 1\). 
        
        If \(\hat{y}^i_v > 0\) for some \(i \in \{j, \ldots, j+l\}\), then lower \(\hat{y}^i_v\) by \(\min\{\hat{y}^i_v, 1 - \hat{\gamma}^i_v\}\) (this does not affect feasibility as \(\hat{y}^i_v\) is contained in only this one constraint, and it improves the objective). If now \(\hat{\gamma}^i_v + \hat{y}^i_v = 1\), then we continue with a smaller range, otherwise we continue with the same range. So, we can assume that for all \(i \in \{j, \ldots, j+l\}\), we have \(\hat{y}^i_v = 0\) and \(\hat{\gamma}^i_v > 1\). Let \(\eta > 0\) be small enough such that if we set \((\hat{\gamma}^i_v)' = \hat{\gamma}^i_v - \eta\) (this does not change the objective) for all \(i \in \{j, \ldots, j+l\}\), we still have  \((\hat{\gamma}^i_v)' \geq 1\).

        For \(i \in \{j, \ldots, j+l\}\) and \(u \in V_t\) such that \(uv \in E_i\), we have
        \begin{equation*}
            \hat{\gamma}^i_u - (\hat{\gamma}^i_v)'
            = \hat{\gamma}^i_u - (\hat{\gamma}^i_v - \eta)
            = \hat{\gamma}^i_u - \hat{\gamma}^i_v + \eta
            \geq \eta
            > 0.
        \end{equation*}
        For \(i \in \{j, \ldots, j+l-1\}\), we have
        \begin{equation*}
            (\hat{\gamma}^i_v)' - (\hat{\gamma}^{i+1}_v)' + \hat{\delta}^i_v
            = (\hat{\gamma}^i_v - \eta) - (\hat{\gamma}^{i+1}_v - \eta) + \hat{\delta}^i_v
            = \hat{\gamma}^i_v - \hat{\gamma}^{i+1}_v + \hat{\delta}^i_v
            \geq 0,
        \end{equation*}
        and
        \begin{equation*}
            (\hat{\gamma}^{i+1}_v)' - (\hat{\gamma}^i_v)' + \hat{d}^i_v
            = (\hat{\gamma}^{i+1}_v - \eta) - (\hat{\gamma}^i_v - \eta) + \hat{d}^i_v
            = \hat{\gamma}^{i+1}_v - \hat{\gamma}^i_v + \hat{d}^i_v
            \geq 0.
        \end{equation*}
        If \(v \in V_{i-1}\), then \(\hat{\gamma}^{i-1}_v + \hat{y}^{i-1}_v = 1\), and so
        \begin{equation*}
            \hat{\gamma}^{i-1}_v - (\hat{\gamma}^i_v)' + \hat{\delta}^{i-1}_v
            = \hat{\gamma}^{i-1}_v - (\hat{\gamma}^i_v - \eta) + \hat{\delta}^{i-1}_v
            = \hat{\gamma}^{i-1}_v - \hat{\gamma}^i_v + \hat{\delta}^{i-1}_v + \eta
            \geq \eta 
            > 0,
        \end{equation*}
        and
        \begin{equation*}
        \begin{split}
            (\hat{\gamma}^i_v)' - \hat{\gamma}^{i-1}_v + \hat{\delta}^{i-1}_v
            &= (\hat{\gamma}^i_v - \eta) - (1 - \hat{y}^{i-1}_v) + \hat{\delta}^{i-1}_v\\
            &= (\hat{\gamma}^i_v - \eta - 1) + \hat{y}^{i-1}_v + \hat{\delta}^{i-1}_v
            \geq 0 + 0 + 0 
            = 0.
        \end{split}
        \end{equation*}
        If \(v \in V_{j+l+1}\), then \(\hat{\gamma}^{j+l+1}_v + \hat{y}^{j+l+1}_v = 1\), and so
        \begin{equation*}
        \begin{split}
            (\hat{\gamma}^{j+l}_v)' - \hat{\gamma}^{j+l+1}_v + \hat{\delta}^{j+l}_v
            &= (\hat{\gamma}^{j+l}_v - \eta) - (1 - \hat{y}^{j+l+1}_v) + \hat{\delta}^{j+l}_v \\
            &= (\hat{\gamma}^{j+l}_v - \eta - 1) + \hat{y}^{j+l+1}_v + \hat{\delta}^{j+l}_v
            \geq 0 + 0 + 0 
            = 0,
        \end{split}
        \end{equation*}
        and
        \begin{equation*}
        \begin{split}
            \hat{\gamma}^{j+l+1}_v - (\hat{\gamma}^{j+l}_v)' + \hat{\delta}^{j+l}_v
            &= \hat{\gamma}^{j+l+1}_v - (\hat{\gamma}^{j+l}_v - \eta) + \hat{\delta}^{j+l}_v\\
            &= \hat{\gamma}^{j+l+1}_v - \hat{\gamma}^{j+l}_v + \hat{\delta}^{j+l}_v + \eta
            \geq \eta 
            > 0.
        \end{split}
        \end{equation*}
        So we find that the solution with \(\hat{\gamma}^i_v\) replaced by \((\hat{\gamma}^i_v)'\) for all \(i \in \{j, \ldots, j+l\}\) is feasible. By possibly repeating this argument, we can decrease all \(\hat{\gamma}^i_v\) to \(1\), such that \(\hat{\gamma}^i_v + \hat{y}^i_v = 1\) holds for all \(i \in \{1, \ldots, k\}\) and \(v \in V_i \cap V_s\).

        By similar arguments we can show that we can also assume without loss of generality that \(-\hat{\gamma}^i_v + \hat{y}^i_v = 0\) for \(i = 1, \ldots, k\), \(v \in V_i \cap V_t\).
    \end{proof}

    Now we can map an optimal solution for \eqref{eq:multistage_flow_LP_dual} to a solution for \eqref{eq:multistage_LP_eps} with the same objective value, as described above. This solution is optimal for \eqref{eq:multistage_LP_eps}, as otherwise we could find a better solution, map the better solution back to a solution for \eqref{eq:multistage_flow_LP_dual} with the same objective value, contradicting the optimality of the starting solution.
\end{proof}

\begin{lemma}
    \label{lem:multistage_eps_to_multistage}
    An integral optimal solution for \eqref{eq:multistage_LP_eps} is also an integral optimal solution for \eqref{eq:multistage_LP}.
\end{lemma}
\begin{proof}
    Let \(\hat{\gamma} = (\hat{y}^i \text{ for } i = 1, \ldots, k, \hat{\delta}^i \text{ for } i = 1, \ldots, k-1, \hat{d}^i \text{ for } i = 1, \ldots, k-1)\) be an integral optimal solution for \eqref{eq:multistage_LP_eps}.
    
    Suppose that \(1^\top \hat{y}^i > \nu(G_i)\) for some \(i = 1, \ldots, k\), then because \(\hat{y}^i\) is integral, \(1^\top \hat{y}^i \geq \nu(G) + 1\). Now replace \(\hat{y}^i\) by a minimum vertex cover \(\widetilde{y}^i\), i.e., \(1^\top \widetilde{y}^i = \nu(G_i)\), and set \(\widetilde{\delta}^i\) and \(\widetilde{d}^i\) accordingly.
    We will have \(\widetilde{y}^i \leq 1\) and \(\hat{y}^{i+1} \leq 1\), and hence also \(\widetilde{\delta}^i \leq 1\) and \(\widetilde{d}^i \leq 1\) (this also holds for superscripts other than \(i\)). Therefore,
    \begin{equation*}
    \begin{split}
        \varepsilon \left( \sum_{i=1}^k \sum_{v \in V_i} \alpha^i_v \hat{y}^i_v + \sum_{i=1}^{k-1} \sum_{v \in V_i \cap V_{i+1}} \beta^i_v \widetilde{\delta}^i_v + b^i_v \widetilde{d}^i_v \right)
        &\leq \varepsilon \left( \sum_{i=1}^k \sum_{v \in V_i} \alpha^i_v + \sum_{i=1}^{k-1} \sum_{v \in V_i \cap V_{i+1}} \beta^i_v + b^i_v  \right)\\
        &\leq \varepsilon \left( \sum_{i=1}^k \sum_{v \in V_i} |\alpha^i_v| + \sum_{i=1}^{k-1} \sum_{v \in V_i \cap V_{i+1}} \beta^i_v + b^i_v \right)\\
        &< \varepsilon \left( 1+ \sum_{i=1}^k \sum_{v \in V_i} |\alpha^i_v| + \sum_{i=1}^{k-1} \sum_{v \in V_i \cap V_{i+1}} \beta^i_v + b^i_v \right)\\
        &= 1,
    \end{split}
    \end{equation*}
    which means we obtain a strictly better solution, contradicting that \(\hat{\gamma}\) is optimal. So, \(1^\top \hat{y}^i = \nu(G_i)\). Hence, \(\hat{\gamma}\) is feasible for \eqref{eq:multistage_LP}.

    Suppose \(\hat{\gamma}\) is not optimal for \eqref{eq:multistage_LP}. Let \(\widetilde{\gamma}\) be an optimal solution for \eqref{eq:multistage_LP}. Then
    \begin{equation*}
        \sum_{k=1}^k \sum_{v \in V_i} \alpha^i_v \widetilde{y}^i_v+ \sum_{i=1}^{k-1} \sum_{v \in V_i \cap V_{i+1}} \beta^i_v \widetilde{\delta}^i_v + b^i_v \widetilde{d}^i_v
        < \sum_{k=1}^k \sum_{v \in V_i} \alpha^i_v \hat{y}^i_v+ \sum_{i=1}^{k-1} \sum_{v \in V_i \cap V_{i+1}} \beta^i_v \hat{\delta}^i_v + b^i_v \hat{d}^i_v
    \end{equation*}
    and \(1^\top \widetilde{y}^i = \nu(G_i) = 1^\top \hat{y}^i\) for all \(i = 1, \ldots, k\). It follows that
    \begin{multline*}
        \sum_{i=1}^k 1^\top \widetilde{y}^i + \varepsilon \left( \sum_{k=1}^k \sum_{v \in V_i} \alpha^i_v \widetilde{y}^i_v+ \sum_{i=1}^{k-1} \sum_{v \in V_i \cap V_{i+1}} \beta^i_v \widetilde{\delta}^i_v + b^i_v \widetilde{d}^i_v \right) \\
        < \sum_{i=1}^k 1^\top \hat{y}^i + \varepsilon \left( \sum_{k=1}^k \sum_{v \in V_i} \alpha^i_v \hat{y}^i_v+ \sum_{i=1}^{k-1} \sum_{v \in V_i \cap V_{i+1}} \beta^i_v \hat{\delta}^i_v + b^i_v \hat{d}^i_v \right),
    \end{multline*}
    because \(\varepsilon > 0\). This contradicts the optimality of \(\hat{\gamma}\) for \eqref{eq:multistage_LP_eps}, hence \(\hat{\gamma}\) must be optimal for \eqref{eq:multistage_LP}.
\end{proof}

Finally, \Cref{lem:multistage_int_polyhedron,lem:flow_dual_to_multistage_eps,lem:multistage_eps_to_multistage} prove \Cref{thm:multistage_int_solution}.

\bibliographystyle{plain}
\bibliography{bib}

\end{document}